%% file: csr.tex
\documentclass{llncs}
\usepackage{fullpage}
\usepackage{graphicx}
\usepackage{psfrag}
\usepackage{amsmath}
\usepackage{ulem}
\usepackage{color}
\usepackage{algorithm}
\usepackage[noend]{algpseudocode}
\usepackage{xspace}
\usepackage{amsfonts}
\usepackage{hyperref}
\usepackage{tikz}
\usepackage{multirow}
\usepackage{url}
\usetikzlibrary{fit, shapes, positioning}

\newcommand{\HH}{\mathcal{H}}
\newcommand{\RR}{\mathcal{R}}

\title{Popularity in the generalized Hospital Residents setting}

\author{Meghana Nasre and Amit Rawat}
\institute{Indian Institute of Technology Madras, India}

\begin{document}
\maketitle

\input{main-file}

\end{document}

%% file: main-file.tex
\input{abstract}
\input{intro}

\noindent {\bf Organization:} 
In Section~\ref{sec:stable-pop-lcsm} we define the
notion of popularity, in Section~\ref{sec:sec2} we present the structural characterization of popular matchings.
In Section~\ref{sec:max-card-pop-lcsm} we describe our algorithms to compute a
maximum cardinality popular matching, and a popular matching amongst maximum
cardinality matchings. 
We conclude with a short discussion about popular matchings in the LCSM
problem.

\input{notion-of-pop}

\input{struct-char}

\input{max-card-pop-trimmed}

\input{pop-amongst-max}

\bibliographystyle{plain}
 \bibliography{references-trimmed}

\appendix
\section{Appendix}
\input{appendix_reduction_pcsm_spa}

%% file: abstract.tex
\begin{abstract}
We consider the problem of computing {\it popular} matchings in a bipartite graph $G = (\RR \cup \HH, E)$ where
$\RR$ and $\HH$ denote a set of residents and a set of hospitals respectively. Each hospital $h$ has a positive capacity
denoting the number of residents that can be matched to $h$. The residents and the hospitals
specify strict preferences over each other.
This is the well-studied Hospital Residents (HR) problem which is a generalization 
of the Stable Marriage (SM) problem.
The goal is to assign residents to hospitals {\it optimally} while respecting the capacities of the hospitals.
Stability is a well-accepted notion of optimality in such problems. However, motivated by the need for larger cardinality matchings, 
alternative notions of optimality like {\it popularity}
have been investigated in the  SM setting. In this paper,
we consider a generalized HR setting -- namely the Laminar Classified Stable Matchings (LCSM$^+$) problem. Here, additionally,
 hospitals can specify classifications over residents in their preference lists and classes have upper quotas. 
We show the following new results:
\noindent  We define a notion of popularity and give a structural characterization of popular matchings for the LCSM$^+$ problem.
 Assume $n = |\RR| + |\HH|$ and $m = |E|$. We give an $O(mn)$ time algorithm for computing a maximum cardinality popular matching in an LCSM$^+$ instance.
 We give an $O(mn^2)$ time algorithm for computing a matching that is popular amongst the maximum cardinality matchings in an LCSM$^+$ instance.
\end{abstract}

%% file: intro.tex
\section{Introduction}
\label{sec:intro}
Consider an academic institution where students credit an elective
course from a set of available courses.
Every student and every course rank a subset of elements from the other set in a
strict order of preference.
Each course has a quota denoting the maximum number of students
it can accommodate.
The goal is to allocate  to every student at most one course respecting the preferences.
 This is the well-studied Hospital Residents
problem~\cite{GS62}. 
We consider its generalization  where, in addition, a course
can {\it classify} students --
for
example, the students may be classified as under-graduates and post-graduates and department-wise and
so on. Depending on the classifications, a student may belong to multiple classes.
Apart from the total quota, each course now has a quota for every class. 
An allocation, in this setting, has to additionally respect the class quotas.
This is the Classified Stable Matching problem introduced by Huang~\cite{Huang2010}.

{\it Stability} is a de-facto notion of optimality in settings
where both set of participants have preferences. Informally,
an allocation of students to courses is  stable if no unallocated student-course pair
has incentive to deviate from the allocation. 
Stability is appealing for several reasons -- stable allocations
are guaranteed to exist, they are efficiently computable and all stable allocations leave the same
set of students unallocated~\cite{GI89}. However, it is known~\cite{Kavitha14} that the cardinality
of a stable allocation can be half the size of the largest sized allocation possible. Furthermore,
in applications like student-course allocation, leaving a large number of students unallocated is
undesirable. Thus, it is interesting to consider notions of optimality which respect preferences 
but possibly compromise stability in the favor of cardinality. Kavitha and Huang~\cite{HuangK11,Kavitha14}
investigated this in the Stable Marriage (SM) setting 
where they
considered {\it popularity} as an alternative to stability. At a high level, an allocation of students to courses
is {\it popular} if no {\it majority} wishes to deviate from the allocation. Here, 
we consider { popularity} in the context of two-sided preferences and one-sided
capacities with classifications.

We formally define our problem now -- we use the familiar hospital residents notation.
Let $G = (\RR \cup \HH, E)$ be a bipartite graph where $|\RR \cup \HH| = n$ and $|E| = m$.
Here $\RR$ denotes the set of residents, $\HH$ denotes the set of hospitals
and every hospital $h \in \HH$ has an upper quota  $q(h)$ denoting the
maximum number of residents $h$ can occupy.
A pair $(r, h) \in E$ denotes that $r$ and $h$ are mutually acceptable to each other.
Each resident (resp. hospital) has a strict ordering of a subset
of the hospitals (resp. residents) that are acceptable to him or her (resp. it).
This ordering is called the preference list of a vertex.
An assignment (or a matching) $M$ in $G$ is
a subset of $E$ such that every resident is assigned to at most one hospital and
a hospital $h$ is assigned at most $q(h)$ residents.
Let $M(r)$ (resp. $M(h)$) denote the hospital (resp. the set of residents)
which are assigned to $r$ (resp. $h$) in $M$.
A hospital $h$ is  under-subscribed if $|M(h)| < q(h)$.
A matching $M$  is {\it stable} if no unassigned pair $(r, h)$ wishes to deviate from $M$.
The goal is to compute a stable matching in $G$.
We denote it by HR$^+$ throughout the paper~\footnote{\scriptsize{We use HR$^+$ instead of HR
for consistency with other problems discussed in the paper.}}.
The celebrated
deferred acceptance algorithm by Gale and Shapley~\cite{GS62} proves that every instance of the HR$^+$ problem
admits a stable matching.

A generalization of the HR$^+$ problem is the Laminar Classified Stable Matching (LCSM) problem 
introduced by Huang~\cite{Huang2010}.
An instance of the LCSM$^+$ problem is an instance of the HR$^+$ problem where additionally,
each hospital $h$ is allowed
to specify a classification over the set of residents in its preference list.
A class $C^h_k$ of a hospital $h$ is a subset of residents in its preference list and
has an associated upper quota $q(C^h_k)$ denoting the maximum number of residents that can be matched to $h$ in $C^h_k$.
(In the LCSM problem~\cite{Huang2010}, classes can have lower quotas as well.)
We assume that the classes of a hospital form a {\it laminar} set. That is, for any two classes $C^h_j$ and $C^h_k$,
either the two classes are disjoint ($C^h_j \cap C^h_k  = \emptyset$),
or one is contained inside the other ($C^h_j \subset C^h_k$ or $C^h_k \subset C^h_j$). 
Huang suitably modified the classical definition of stability to account for the presence of these classifications.
He showed that every instance of the LCSM$^+$ problem admits a stable matching
which can be computed in $O(mn)$ time~\cite{Huang2010}.
A restriction of the LCSM$^+$ problem, denoted by Partition Classified Stable Matching (PCSM$^+$), is where
the classes of every hospital partition the residents in its preference list.

Motivated by the need to output larger cardinality matchings, we consider
computing {\it popular} matchings in the LCSM$^+$ problem. The notion of popularity uses {\it votes} to compare two matchings. Before we can define voting in the LCSM$^+$ setting,
it is useful to discuss voting in
the context of the SM problem.

\noindent {\bf Voting in the SM setting:}
Let $G = (\RR \cup \HH, E)$ be an instance of the SM problem and let $M$ and $M'$ be any two matchings in $G$.
A vertex $u \in \RR \cup \HH$ (where each hospital $h$ has $q(h) = 1$)
 prefers $M$ over $M'$ and therefore votes for $M$ over $M'$ if either
(i) $u$ is matched in $M$ and unmatched in $M'$ or
(ii) $u$ is matched in both $M$ and $M'$ and prefers $M(u)$ over $M'(u)$.
A matching $M$ is more popular than $M'$ if the number of votes that $M$ gets as compared to $M'$ is
greater than the number of votes that $M'$ gets as compared to $M$. A matching $M$ is popular
if there does not exist any matching that is more popular than $M$.
In the SM setting it is known that a stable matching is popular, however it was shown to be {\it minimum}
cardinality popular matching~\cite{HuangK11}.
Huang and Kavitha~\cite{HuangK11,Kavitha14} gave efficient algorithms
for computing a max-cardinality popular matching and a popular matching amongst max-cardinality matchings in an SM instance.

\noindent {\bf Voting in the capacitated setting:}
To extend voting in the capacitated setting, 
we assign a hospital $h$ as many votes as its upper
quota $q(h)$. This models the scenario in which hospitals with larger capacity get a larger share of votes.
For the HR$^+$ problem, a hospital $h$ compares the most
preferred resident in $M(h) \setminus M'(h)$ to the most preferred resident in $M'(h) \setminus M(h)$ (and
votes for $M$ or $M'$ as far as those two residents are concerned) and so on.
For this voting scheme, we can obtain analogous results for computing popular matchings
in the HR$^+$ problem via the standard technique of {\it cloning} (that is, 
creating $q(h)$ copies of a
hospital $h$ and appropriately modifying preference lists of the residents and hospitals
\footnote {\scriptsize{For every hospital in the cloned graph, its preference list is the same as
in the original instance.
For every hospital $h$, fix an ordering of its clones. The preference
list of a resident $r$ in the cloned instance is obtained by replacing the occurrence of $h$ by
the fixed ordering of its clones.
We refer the reader to ~\cite{BrandlK16,Kiraly11} for details.}}).
However, our interest is in the LCSM$^+$ problem, for which we are not aware of any reduction to the SM problem.
Furthermore, we show
that the straightforward voting scheme as defined in the HR$^+$ does not suffice for the LCSM$^+$ problem.
Therefore, we define a voting scheme for a hospital which takes into consideration the classifications
as well as ensures that every stable matching in the LCSM$^+$ instance is popular.
We show the following results: 
\begin{itemize}
\item 
We define a notion of popularity for the LCSM$^+$ problem.
Since our definition ensures that stable matchings are popular --
this guarantees the existence of popular matchings in the LCSM$^+$ problem.
\item 
We give a characterization of popular matchings for the LCSM$^+$ problem, 
which is a natural extension of the 
characterization of popular matchings in SM setting~\cite{HuangK11}.
\item 
We obtain the following algorithmic results.
An $O(m+n)$ (resp. $O(mn))$ time algorithm  for computing a maximum cardinality popular matching in a PCSM$^+$ (resp. LCSM$^+$) instance.
 An $O(mn)$ (resp. $O(mn^2)$) time algorithm for computing a popular matching amongst maximum cardinality matchings in a
PCSM$^+$ (resp. LCSM$^+$) instance.
\end{itemize}

Very recently, independent of our work, two different groups~\cite{BrandlK16,Kamiyama16}
have considered popular matchings in the one-to-many setting. Brandl and Kavitha~\cite{BrandlK16} have
considered computing {\it popular} matchings in the HR$^+$ problem.
In their work as well as ours, a hospital $h$ is assigned as many votes as
its capacity to compare two matchings $M$ and $M'$.
In contrast, by the definition of popularity in \cite{BrandlK16}, a hospital $h$
chooses the most adversarial ordering of residents in $M(h) \setminus M'(h)$
and $M'(h) \setminus M(h)$  for comparing $M$ and $M'$.
However, it is interesting to note that in an HR$^+$ instance the same matching is output by both our algorithms.
On the other hand, we remark that the model considered
in our paper is a more general one than the one considered in \cite{BrandlK16}.
Kamiyama~\cite{Kamiyama16} has generalized our work and the results in \cite{BrandlK16}
using a matroid based approach.


We finally remark that one can consider voting schemes where a hospital is given a {\it single} vote instead
of capacity many votes. In one such scheme, a hospital compares the set of residents in $M(h)$ and $M'(h)$ in
lexicographic order and votes accordingly. However, when such a voting is used, it is possible to construct instances
where a stable matching is {\it not popular}. The techniques in this paper use the fact that stable matchings
are popular, therefore it is unclear if our techniques will apply for such voting schemes.

\noindent{\bf Related Work:}
The notion of popularity was introduced by G\"{a}rdenfors~\cite{Gar75} in the context of stable matchings.
In \cite{AIKM07} Abraham~et~al. studied popularity in the one-sided preference list model.
As mentioned earlier, our work is inspired by a series of papers where popularity
is considered as an alternative to stability in the stable marriage setting
by Huang, Kavitha and Cseh \cite{CsehK16,HuangK11,Kavitha14}.
Bir\'o~et~al.~\cite{BiroMM10} give several practical scenarios where stability
may be compromised in the favor of size. 
The PCSM$^+$ problem is a special case of the
Student Project Allocation (SPA) problem studied by Abraham~et~al.~\cite{Abraham2007}.
They gave a linear time algorithm to compute a stable matching in an instance of the SPA problem.
In this paper, we use the algorithms of Abraham~et~al.~\cite{Abraham2007} and Huang~\cite{Huang2010}
for computing stable matchings in the PCSM$^+$ and LCSM$^+$ problems.
Both these algorithms follow the standard {\it deferred acceptance} algorithm
of Gale and Shapley with problem specific modifications.
We refer the reader to \cite{Abraham2007} and \cite{Huang2010} for
details.

%% file: notion-of-pop.tex
\section{Stability and popularity in the LCSM$^+$ problem}
\label{sec:stable-pop-lcsm}
Consider an instance $G = (\RR \cup \HH, E)$ of the LCSM$^+$ problem.
As done in \cite{Huang2010}, assume that for every $h \in \HH$ there is a class $C^h_{*}$ containing
all the residents in the preference list of $h$ and  $q(C^h_{*}) = q(h)$. For a hospital $h$, let $T(h)$ denote
the tree of classes corresponding to $h$ where $C^h_{*}$ is the root of $T(h)$. The leaf classes
in $T(h)$ denote the most refined classifications for a resident whereas as we move up in the tree from a leaf node to the root,
the classifications gets coarser.

To define stable matchings in the LCSM problem, Huang
introduced the notion of a {\it blocking group} w.r.t. a matching.
Later, Fleiner and Kamiyama~\cite{FleinerK2012} defined a blocking pair which is equivalent
to a blocking group of Huang. 
We use the definition of stability from \cite{FleinerK2012} which we recall below. 
A set $S = \{ r_1, \ldots, r_l \}$ is  {\it feasible} for a hospital $h$
if $|S| \le q(h)$ and for every class $C^h_j$ of $h$ (including the root class  $C^h_*$), we have
$|C_j^h \cap S| \le q(C_j^h)$.
A matching $M$ in $G$ is feasible if every resident is matched
to at most one hospital, and $M(h)$ is feasible for every hospital $h \in \mathcal{H}$.
A pair $(r, h) \notin M$ blocks $M$ iff both the conditions below hold:
\begin{itemize}
    \item $r$ is unmatched in $M$, or $r$ prefers $h$ over $M(r)$, and
    \item either the set $M(h) \cup \{r\}$ is feasible for $h$, or
        there exists a resident $r' \in M(h)$, such that $h$ prefers
        $r$ over $r'$, and $(M(h) \setminus \{r'\}) \cup \{r\}$ is feasible for $h$.
\end{itemize}

\noindent A feasible matching $M$ in $G$ is stable if $M$ does not admit any blocking pair.
\subsection{Popularity}
To define popularity, we need to specify how a hospital compares two sets $M(h)$ and $M'(h)$ in an LCSM$^+$ setting, where $M$ and $M'$ are
two feasible matchings in the instance.
\input{example.tex}

We formalize the above observations in the rest of the section.
To take into account the classifications, for a hospital $h$ and the matchings $M$ and $M'$,
we set up a correspondence between residents in $M(h) \setminus M'(h)$
and the residents in $M'(h) \setminus M(h)$.
That is, we define:  
\begin{eqnarray*}
{\bf corr}: M(h) \oplus M'(h) \rightarrow M(h) \oplus M'(h) \cup \{\bot\}
\end{eqnarray*}
For a resident $r \in M(h) \oplus M'(h)$ we denote by ${\bf corr}(r)$ the
corresponding resident to which $r$ gets compared when the hospital $h$ casts its votes. 
We let ${\bf corr}(r) = \bot$ if $r$ does not have a corresponding resident to
be compared to from the other matching.
The pseudo-code for the algorithm to compute the  {\bf corr} function is given below.
\input{corr-modified}

The algorithm begins by setting {\bf corr} for every $r \in M(h) \oplus M'(h)$ to $\bot$. The algorithm
maintains two sets of residents $Y = M(h) \setminus M'(h) $ and $Y' = M'(h) \setminus M(h)$ for whom ${\bf corr}$ needs to be set.
As long as the sets $Y$ and $Y'$ are both non-empty, the algorithm repeatedly computes
for every class $C^h_j$ (including the root class $C^h_*$) the sets $X_j = C^h_j \cap Y$ and $X'_j = C^h_j \cap Y'$.
The algorithm then chooses one of the most refined classes, say $C^h_f$ in $T(h)$, for whom $X_f$ and $X'_f$ are both non-empty.
Finally, residents in $X_f$ and $X_f'$ are sorted according to the preference ordering of $h$ and
the ${\bf corr }$ of the $k$-th most preferred resident in $X_f$ is set to the
$k$-th most preferred resident in $X_f'$, where $k = 1, \ldots, \min\{|X_f|, |X_f'|\}$.

For $r \in \RR$, and any feasible matching $M$ in $G$, if $r$ is unmatched in $M$ then,
$M(r) = \bot$.
A vertex prefers any of its neighbours over $\bot$.
For a vertex $u \in \RR \cup \HH$, let $x, y \in N(u) \cup \{\bot\}$,
where $N(u)$ denotes the neighbours of $u$ in $G$.
\begin{eqnarray*}
vote_u(x, y) &=& +1 \ \ \ \ \mbox { if $u$ prefers  $x$ over $y$} \\ 
             &=& -1 \ \ \ \ \mbox { if $u$ prefers $y$ over $x$} \\
             &=& \ \  0 \ \ \ \ \mbox { if $x$ = $y$ }
\end{eqnarray*}
\noindent Using the above notation, the vote of a resident is easy to define --
a resident $r$ prefers $M'$ over $M$ iff the term $\mathcal{V}_r > 0$,
where $\mathcal{V}_r = vote_r(M'(r), M(r))$.

Recall that a hospital $h$ uses $q(h)$ votes to compare $M$ and $M'$.
Let $q_1(h) = |M(h) \cap M'(h)|$  (number of common residents assigned to $h$ in $M$ and $M'$) and 
$q_2(h) = q(h) - \max\{|M(h)|, |M'(h)|\}$ (number of unfilled positions of $h$ in both $M$ and $M'$). 
Our voting scheme ensures that $q_1(h) + q_2(h)$ votes of $h$ remain unused when comparing $M$ and $M'$.
A hospital $h$ prefers $M'$ over $M$ iff the term $\mathcal{V}_h > 0$, where $\mathcal{V}_h$ is defined as follows:
\begin{eqnarray*}
\mathcal{V}_h = (|M'(h)| - |M(h)|) +
\sum_{\substack {r \in M'(h) \setminus M(h) \\  { \&\& }\\ {\bf corr}(r) \neq \bot}} vote_h(r, {\bf corr}(r))
\end{eqnarray*}
The first term in the definition of $\mathcal{V}_h$ counts the votes of $h$  w.r.t. the residents from either $M$ or $M'$
that did not find correspondence. The second term counts the votes of $h$ w.r.t. the residents each of
which has a
corresponding resident from the other matching.
We note that in the SM setting, ${\bf corr}(r)$ will simply be $M(h)$. Thus,
our definition of votes in the presence of capacities is a natural generalization of the voting scheme in the SM problem.

Let us define the term $\Delta (M', M)$ as the difference between the votes that $M'$ gets over $M$
and the votes that $M$ gets over $M'$.
\begin{eqnarray*}
\Delta (M', M) = \sum_{r \in \RR} \mathcal{V}_r + \sum_{h \in \HH}\mathcal{V}_{h}
\end{eqnarray*}
\begin{definition}
A matching $M$ is popular in $G$ iff for every feasible matching $M'$, we have $\Delta(M', M) \le 0$.
\end{definition}
\subsection{Decomposing $M \oplus M'$}
Here, we present a simple algorithm which allows us to decompose edges of components of $M \oplus M'$ in
an instance into alternating paths and cycles. 
Consider the graph  $\tilde{G} = (\RR \cup \HH, M \oplus M')$, for any two feasible matchings $M$ and $M'$ in $G$.
We note that the degree of every resident in $\tilde{G}$ is at most 2 and
the degree of every hospital in $\tilde{G}$ is at most $2 \cdot q(h)$.
Consider any connected component $\mathcal{C}$ of $\tilde{G}$
and let $e$ be any edge in $\mathcal{C}$. We observe that it is possible
to construct a unique maximal $M$ alternating path or cycle $\rho$ containing
$e$ using the following simple procedure. Initially $\rho$ contains only the edge $e$.
\begin{enumerate}
\item Let $r \in \RR$ be one of the end points of the path $\rho$, and assume that $(r, M(r)) \in \rho$. 
We grow $\rho$ by adding the edge $(r, M'(r))$. 
Similarly if an edge from $M'$ is incident on $r$ in $\rho$, we grow the path by adding the edge $(r, M(r))$ if it exists.

\item Let $h \in \HH$ be one of the end points of the path $\rho$, and assume that $(r, h) \in M \setminus M'$
belongs to $\rho$.
We extend $\rho$ by adding $({\bf corr}(r), h)$ if ${\bf corr}(r)$ is not equal to $\bot$.
A similar step is performed if the last edge on $\rho$ is $(r, h) \in M' \setminus M$.
\item We stop the procedure when we complete a cycle (ensuring that the two adjacent residents of a hospital are ${\bf corr}$ for each other according to the hospital), or the path can no longer be extended.
Otherwise we go to Step~1 or Step~2 as applicable and repeat.
\end{enumerate}
The above procedure gives us a unique decomposition of a connected component in $\tilde{G}$
into alternating paths and cycles. Note that a hospital may appear multiple times in a single path or a cycle
and also can belong to more than one alternating paths and cycles.
\input{app-decom-eg-new}

Let $\mathcal{Y}_{M \oplus M'}$ denote the collection of alternating paths and alternating cycles
obtained by decomposing every component of $\tilde{G}$.

We now state  a useful property about any alternating path or cycle in
$\mathcal{Y}_{M \oplus M'}$.
\begin{lemma}
\label{lem:rho_xor_feasible_g}
If $\rho$ is an alternating path or an alternating cycle in
$\mathcal{Y}_{M \oplus M'}$, then $M \oplus \rho$ is a feasible matching in $G$.
\end{lemma}
\input{lemma1-proof}

\noindent As was done in \cite{Kavitha14}, it is convenient to
label the edges of $M' \setminus M$ and use these labels to compute $\Delta (M', M)$.
Let $(r, h) \in M' \setminus M$; the label on $(r, h)$ is a tuple:
\begin{eqnarray*}
(vote_r(h, M(r)), \hspace{0.1in} vote_h(r, {\bf corr}(r)))
\end{eqnarray*}
Note that since we are labeling edges of $M'\setminus M$,
both entries of the tuple come from the set $\{-1, 1\}$.
With these definitions in place, we are ready to  give the structural characterization of popular matchings
in an LCSM$^+$ instance.



%% file: example.tex
\subsubsection {Illustrative example}
\label{sec:app:ex}
Consider the following LCSM$^+$ instance where
$\RR = \{ r_1, \ldots, r_4 \}$ and
$\HH = \{ h_1, \ldots, h_3\}$ and the preference lists of the residents and hospitals
are as given in Figure~\ref{fig:lcsm_instance}(a) and (b) respectively. The preferences can be read as follows: resident
$r_1$ has $h_1$ as his top choice hospital. Resident $r_2$ has $h_2$ as its top choice hospital followed by $h_1$ which is his second choice hospital and so on.
For $h \in \{h_2, h_3\}$ we have $q(h) = 1$ and both these
hospitals have a single class $C_*^h$ containing all the residents in the preference list of $h$ and
$q(C_*^h) = q(h)$.
For hospital $h_1$ we have $q(h_1) = 2$ and the classes provided by $h_1$
are $C_1^{h_1} = \{ r_1, r_2 \}, C_2^{h_1} = \{ r_3, r_4 \}, C_*^{h_1} = \{ r_1, r_2, r_3, r_4 \}$ with quotas as follows:
$q(C_1^{h_1}) = q(C_2^{h_1}) = 1$ and $q(C_*^{h_1}) = 2$.
We remark that the example in Figure~\ref{fig:lcsm_instance} is also a PCSM$^+$ instance.
Figure~\ref{fig:lcsm_instance}(c) shows the tree $T(h_1)$.

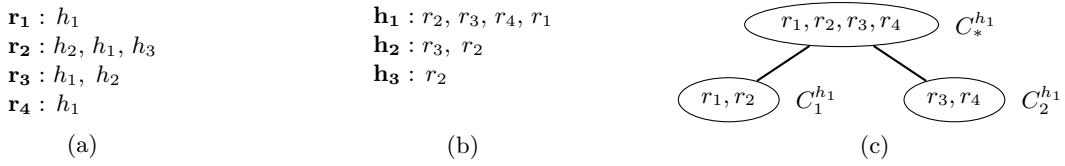
\begin{figure}[ht]
\begin{minipage}{0.3\linewidth}
\begin{center}
\begin{tabular}{cccc}
$\bf{r_1:}$&$ h_1$\\
$\bf{r_2:}$&$ h_2,$&$ h_1,$&$h_3$\\
$\bf{r_3:}$&$ h_1,$&$ h_2$\\
$\bf{r_4:}$&$ h_1$\\
\end{tabular}

\vspace{0.06in}
(a)
\end{center}
\end{minipage}
\begin{minipage}{0.3\linewidth}
\begin{center}
\begin{tabular}{ccccc}
$\bf{h_1:}$ &$r_2,$&$r_3,$& $r_4,$ & $r_1$\\
$\bf{h_2:}$ &$r_3,$&$r_2$& \\
$\bf{h_3:}$ &$r_2$\\
\vspace{0.01in}
\end{tabular}

\vspace{0.06in}
(b)
\end{center}
\end{minipage}
\begin{minipage}{0.35\linewidth}
\begin{center}
\begin{tikzpicture}{h}
    \node (A) [draw, ellipse] at (1.5, 2) {$ r_1, r_2, r_3, r_4 $};
    \node (B) [draw, ellipse] at (0, 1) {$ r_1, r_2 $};
    \node (C) [draw, ellipse] at (3, 1) {$ r_3, r_4 $};

    \node [right=0.1cm of A] {$C^{h_1}_*$};
    \node [right=0.1cm of B] {$C^{h_1}_1$};
    \node [right=0.1cm of C] {$C^{h_1}_2$};

    \draw[black, thick] (A) -- (B);
    \draw[black, thick] (A) -- (C);
\end{tikzpicture}

\vspace{0.01in}
(c)
\end{center}
\end{minipage}
\caption{(a) Resident preferences, (b) Hospital preferences, (c) $T(h_1)$.
The matchings $M =  \{(r_1, h_1), (r_2, h_2), (r_3, h_1)\}$,  $M' = \{(r_2, h_1), (r_3, h_2), (r_4, h_1)\}$,
and $M'' = \{(r_1, h_1), (r_2, h_3), (r_3, h_2), (r_4, h_1) \}$ are all feasible in the instance.}
\label{fig:lcsm_instance}
\end{figure}

Consider the two feasible matchings
$M$ and $M' $ defined in Fig.~\ref{fig:lcsm_instance}.
Note that $M$ is stable in the instance whereas the edge $(r_3, h_1)$ blocks $M'$.
While comparing $M$ and $M'$,
the vote for every vertex $u$ in the instance except $h_1$ is clear --
$u$ compares $M(u)$ with $M'(u)$ and
votes accordingly.
In order for $h_1$ to vote between $M$ and $M'$, the hospital compares between
$M(h_1) = \{r_1, r_3\}$ and $M'(h_1) = \{r_2, r_4\}$.
A straightforward way is to compare $r_3$ with $r_2$
(the most preferred resident in $M(h_1)$ to the most preferred resident in $M'(h_1)$)
and then compare $r_1$ with $r_4$ (second most preferred resident in $M(h_1)$ to second most preferred resident in $M'(h_1)$). Thus,
both the votes of $h_1$ are in favor  of $M'$ when compared with $M$.
Such a comparison has two issues -- (i) it ignores the classifications given by $h_1$,
and (ii) the number of votes that $M'$ gets when compared with $M$ is more than
the number of votes that $M$ gets as compared to $M'$. Therefore $M'$ is more popular than $M$ which
implies that $M$ (a stable matching) is {\bf not}  popular. 

We propose a comparison scheme for hospitals  which addresses both the issues.
In the above example, we note that $r_1 \in M(h)$ has a corresponding resident $r_2 \in M'(h)$ to be compared to
in one of the most refined classes $C_1^{h_1}$ (see Figure~\ref{fig:lcsm_instance}(c)).
Thus, we compare $r_1$ with $r_2$. The resident $r_3 \in M(h)$ is compared to $r_4 \in M(h)$ another leaf class $C_2^{h_1}$.
According to this comparison, $h_1$ is indifferent between
$M$ and $M'$ and $M'$ is no longer more popular than $M$.  Note that, although in the example, both the comparisons happen in a leaf
class, this may not be the case in a general instance.
Finally, we note that the matching $M''$ is a popular matching in the instance and is strictly larger in size than the stable matching $M$.

%% file: corr-modified.tex
\label{sec:app-corr}
\begin{algorithm}[H]
\label{algo:find-correspondence-lcsm}
\begin{algorithmic}[1]
    \Procedure{Find-Correspondence}{$h, M, M'$}
        \State let $T(h)$ be the classification tree associated with $h$
        \State set ${\bf corr}(r) = \bot$ for each $r \in M(h) \oplus M'(h)$
        \State $Y = M(h) \setminus M'(h)$; $Y' = M'(h) \setminus M(h)$
        \While {$ Y \neq \emptyset$ and $Y' \neq \emptyset$}
            \For{each class $C_j^h$ in $T(h)$}
                \State $X_j =C_j^h  \cap Y$
                \State $X_j' = C_j^h \cap Y'$
            \EndFor
	    \State Let $C_f^h$ be one of the most refined classes for which
                $X_f \neq \emptyset$ and $X_f' \neq \emptyset$.
            \For {$k = 1, \ldots, \min(|X_f|, |X_f'|)$}
                    \State let $r$ be the $k$-th most preferred resident in $X_f$
                    \State let $r'$ be the $k$-th most preferred resident in $X'_f$
                    \State set ${\bf corr}(r) = r'$, and ${\bf corr}(r') = r$
		    \State $Y = Y \setminus \{r\}$; $Y' = Y' \setminus \{r'\}$
             \EndFor
        \EndWhile
    \EndProcedure
\end{algorithmic}
\caption{Correspondence between residents of $M(h)$ and $M'(h)$}
\end{algorithm}

%% file: app-decom-eg-new.tex
Figure~\ref{fig:cycle-dec} gives an example of the decomposition of the two
feasible matchings in the instance in Figure~\ref{fig:lcsm_instance}.
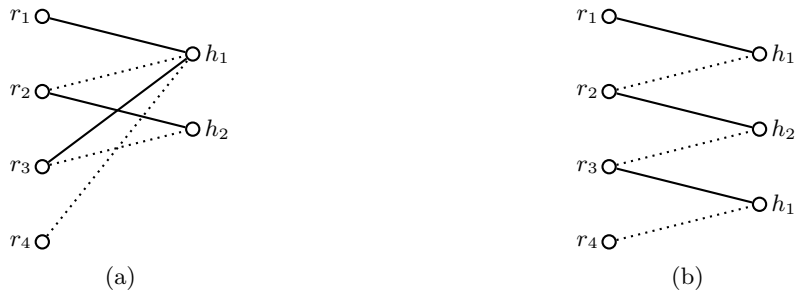
\begin{figure}[h]
\begin{minipage}{0.45\linewidth}
\centering
\begin{tikzpicture}[every node/.style={circle, draw, thick, inner sep=0pt, minimum width=5pt}]
\node (r_1)[label=left:$r_1$]  at (0, 4) {};
\node (r_2)[label=left:$r_2$] at (0, 3) {};
\node (r_3)[label=left:$r_3$] at (0, 2) {};
\node (r_4)[label=left:$r_4$] at (0, 1) {};
\node (h_1)[label=right:$h_1$] at (2, 3.5) {};
\node (h_2)[label=right:$h_2$] at (2, 2.5) {};

\draw[thick]         (r_1) -- (h_1);
\draw[thick]         (r_2) -- (h_2);
\draw[thick]         (r_3) -- (h_1);

\draw[thick, dotted] (r_2) -- (h_1);
\draw[thick, dotted] (r_3) -- (h_2);
\draw[thick, dotted] (r_4) -- (h_1);
\end{tikzpicture}

(a)
\end{minipage}
\begin{minipage}{0.45\linewidth}
\centering
\begin{tikzpicture}[every node/.style={circle, draw, thick, inner sep=0pt, minimum width=5pt}]
\node (r_1)[label=left:$r_1$]  at (0, 4) {};
\node (r_2)[label=left:$r_2$] at (0, 3) {};
\node (r_3)[label=left:$r_3$] at (0, 2) {};
\node (r_4)[label=left:$r_4$] at (0, 1) {};
\node (h_11)[label=right:$h_1$] at (2, 3.5) {};
\node (h_2)[label=right:$h_2$] at (2, 2.5) {};
\node (h_12)[label=right:$h_1$] at (2, 1.5) {};
\draw[thick]         (r_1) -- (h_11);
\draw[thick]         (r_2) -- (h_2);
\draw[thick]         (r_3) -- (h_12);

\draw[thick, dotted] (r_2) -- (h_11);
\draw[thick, dotted] (r_3) -- (h_2);
\draw[thick, dotted] (r_4) -- (h_12);

\end{tikzpicture}

(b)
\end{minipage}
\caption{$M$ and $M'$ are feasible matchings in the example as defined in Fig.~\ref{fig:lcsm_instance}. (a) $\tilde{G} = (\RR \cup \HH, M \oplus M')$; bold edges belong
    to $M$, dashed edges belong to $M'$.
    (b) shows the decomposition of the edges of the component of
    $\tilde{G}$ into a single path.}
\label{fig:cycle-dec}
\end{figure}

%% file: lemma1-proof.tex
\begin{proof}
Let $\langle r', h, r \rangle$ be any sub-path of $\rho$,
where $r' = {\bf corr}(r)$, and $(r, h) \in M$.
We  prove that $(M(h) \setminus \{ r \}) \cup \{ r' \}$ is feasible for $h$.
\begin{figure}[ht]
\centering
\begin{tikzpicture}{h}
    \draw[dotted, black, thick] (3.5, 4) -- (2, 2);
    \draw[black, thick] (3.5, 4) -- (5.8, 1.5);

    \draw[dotted, black, thick] (1, 1) -- (2, 2);
    \draw[dotted, black, thick] (3, 1) -- (2, 2);
    \draw[dotted, black, thick] (2, 0) -- (2, 2);
    \draw[dotted, black, thick] (0, 0) -- (1, 1);
    \draw[dotted, black, thick] (4, 0) -- (3, 1);

    \node[draw,
        fill=white,
        shape border rotate=120,
        regular polygon,
        regular polygon sides=3,
        node distance=2cm,
        minimum height=6em] at (5.4, 1.5) {};

    \fill[white] (3.5, 4) ellipse (0.5 and 0.25);
    \fill[white] (2, 2) ellipse (0.5 and 0.25);
    \fill[white] (2, 0) ellipse (0.5 and 0.25);
    \fill[white] (1, 1) ellipse (0.5 and 0.25);
    \fill[white] (0, 0) ellipse (0.5 and 0.25);
    \fill[white] (3, 1) ellipse (0.5 and 0.25);
    \fill[white] (4, 0) ellipse (0.5 and 0.25);

    \draw[black, thick] (3.5, 4) ellipse (0.5 and 0.25) node[right, xshift=0.5cm] {$C^h_*$};
    \draw[black, thick] (2, 2) ellipse (0.5 and 0.25) node[right, xshift=0.5cm] {$C^h_k$};
    \draw[black, thick] (2, 0) ellipse (0.5 and 0.25);
    \draw[black, thick] (1, 1) ellipse (0.5 and 0.25);
    \draw[black, thick] (0, 0) ellipse (0.5 and 0.25) node[below, yshift=0.2cm] {$.., r, ..$};
    \draw[black, thick] (0, 0) ellipse (0.5 and 0.25) node[below, yshift=-0.3cm] {$C^h_i$};
    \draw[black, thick] (3, 1) ellipse (0.5 and 0.25) node[right, xshift=0.5cm] {$C^h_t$};
    \draw[black, thick] (4, 0) ellipse (0.5 and 0.25) node[below, yshift=0.3cm] {$.., r', ..$};
    \draw[black, thick] (4, 0) ellipse (0.5 and 0.25) node[below, yshift=-0.3cm] {$C^h_j$};
\end{tikzpicture}
\caption{The classification tree T(h) for a hospital $h$.}
\label{fig:tree}
\end{figure}
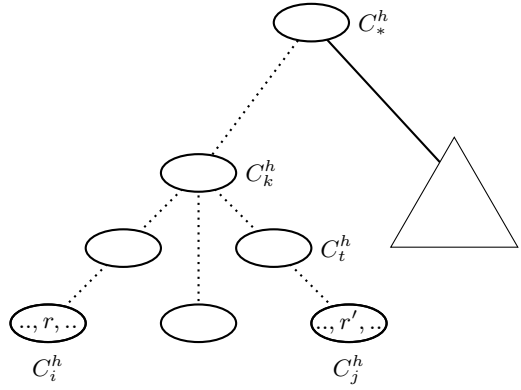
Let $C^h_i$ (resp. $C^h_j$) be the unique leaf class of $T(h)$
containing $r$ (resp. $r'$).
See Figure~\ref{fig:tree}.
We consider the following two cases:
\begin{itemize}
\item $r$ and $r'$ belong to the same leaf class in $T(h)$, i.e. $C^h_i = C^h_j$.
In this case, it is easy to note that $(M(h) \setminus \{r\}) \cup \{r'\}$ is feasible for $h$.
\item $r$ and $r'$ belong to different leaf classes of $T(h)$, i.e. $C^h_i \neq C^h_j$.
Observe that $|(M(h) \setminus \{r\}) \cup \{r'\}|$ can  violate the
upper quota only for those classes of $T(h)$ which contain $r'$ but do not contain $r$. 
Let $C^h_k$ be the least common ancestor of $C^h_i$ and $C^h_j$ in $T(h)$.
It suffices to look at any  class $C^h_t$ which lies in the path from
$C^h_k$ to $C^h_j$ excluding the class $C^h_k$ and show
that $|(M(h)  \cap C^h_t) \cup \{r'\}| \le q(C^h_t)$.
As $r' = {\bf corr}(r)$ and $r \notin C^h_t$, we claim that
$|M(h) \cap C^h_t| < |M'(h) \cap C^h_t| \le q(C^h_t)$. The first inequality is
due to the fact that $r'$ did not find a corresponding resident in the set $(M(h) \setminus M'(h)) \cap C^h_t$.
The second inequality is because $M'$ is feasible.
Thus,
$(M(h) \cap C^h_t) \cup \{ r' \}$ does not violate the upper quota for $C^h_t$.
Therefore $(M(h) \setminus \{ r \}) \cup \{ r' \}$ is feasible for $h$.
\end{itemize}
We note that the hospital $h$ may occur multiple times on $\rho$. Let 
$M(h)_\rho$ denote the set of residents matched to $h$ restricted to $\rho$.
To complete the proof of the Lemma, 
we need to prove that $(M(h) \setminus M(h)_\rho ) \cup M'(h)_\rho$ is feasible for $h$. 
The arguments for this
follow from the arguments given above.
\qed
\end{proof}

%% file: struct-char.tex
\section{Structural characterization of popular matchings}
\label{sec:sec2}


Let $G = (\RR \cup \HH, E)$ be an LCSM$^+$ instance and let
$M$ and $M'$ be two feasible matchings in $G$.
Using the ${\bf corr}$ function, we obtain a correspondence of residents in $M(h) \oplus M'(h)$ for
every hospital $h$ in $G$.
Let $\tilde{G} = (\RR \cup \HH, M \oplus M')$ and let $\mathcal{Y}_{M \oplus M'}$ denote
the collection of alternating paths and cycles obtained by decomposing every component of $\tilde{G}$. 
Finally, we label the edges of $M' \setminus M$ using appropriate votes.
The goal of these steps is to 
is to rewrite the term $\Delta(M', M)$ as a sum of labels on edges.

We note that the only vertices for whom their vote does not get captured on the edges of
$M' \setminus M$ are vertices that are matched in $M$ but not matched in $M'$.
Let $\mathcal{U}$ denote the multi-set of vertices that are end points of paths in $\mathcal{Y}_{M \oplus M'}$
such that there is no $M'$ edge incident on them. Note that the same hospital can
belong to multiple alternating paths and cycles in $\mathcal{Y}_{M \oplus M'}$,
therefore we need a multi-set. All vertices in $\mathcal{U}$ prefer $M$ over $M'$
and hence we add a $-1$ while capturing their vote in $\Delta(M', M)$.
We can write $\Delta(M', M)$ as:
\begin{eqnarray*}
 \Delta(M',M) =
   \sum_{x \in \mathcal{U} } -1  +
  \sum_{\rho \in \mathcal{Y}_{M \oplus M'}} \left( \sum_{\substack{(r,h) \in (M' \cap \rho)
					  }} \{vote_r(h, M(r))    + vote_h(r, {\bf corr}(r))\} \right)
\end{eqnarray*}
We now delete the edges labeled $(-1, -1)$ from all
paths and cycles $\rho$ in $\mathcal{Y}_{M \oplus M'}$. This
simply breaks paths and cycles into one or more paths.
Let this new collection of paths and cycles be denoted by $\tilde{\mathcal{Y}}_{M \oplus M'}$.
Let $\tilde{\mathcal{U}}$ denote the multi-set of vertices that are end points of paths in $\tilde{\mathcal{Y}}_{M \oplus M'}$
such that there is no $M'$ edge incident on them.
We rewrite $\Delta(M', M)$ as:
\begin{eqnarray*}
 \Delta(M',M) =
   \sum_{x \in \tilde{\mathcal{U}} } -1  +
  \sum_{\rho \in \tilde{\mathcal{Y}}_{M \oplus M'}} \left( \sum_{\substack{(r,h) \in (M' \cap \rho)
					  }} \{vote_r(h, M(r))    + vote_h(r, {\bf corr}(r))\} \right)
\end{eqnarray*}
Theorem below characterizes a popular matching.
\begin{theorem}
\label{thm:charc}
A feasible matching ${M}$ in ${G}$ is popular iff for any feasible matching $M'$ in $G$,
the set $\tilde{\mathcal{Y}}_{M \oplus M'}$ does not contain any of the following:
\begin{enumerate}
\item An alternating cycle with a  $(1, 1)$ edge,
\item An alternating path which has a $(1, 1)$ edge and starts with an unmatched
    resident in $M$ or a hospital which is under-subscribed in $M$.
\item An alternating path which has both its ends matched in $M$ and has two or more $(1, 1)$ edges.
\end{enumerate}
\end{theorem}

\input{appendix_proof_struct_char}

We now prove that every stable matching in an LCSM$^+$ instance is popular.
\begin{theorem}
Every stable matching in an LCSM$^+$ instance $G$ is popular.
\label{lcsm-stable-popular-proof}
\end{theorem}
\begin{proof} 
Let $M$ be a stable matching in $G$. 
For any feasible matching $M'$ in $G$ consider the set 
$\mathcal{Y}_{M \oplus M'}$.
To prove that $M$ is stable, 
it suffices to show that there does not exist a path or cycle $\rho \in \mathcal{Y}_{M \oplus M'}$
such that an edge of $\rho$ is labeled $(1, 1)$.
For the sake of contradiction, assume that  $\rho$ is such a path or cycle, which has an edge
$(r', h) \in M' \setminus M$ labeled $(1, 1)$.
Let $r = {\bf corr}(r')$, where $(r, h) \in M \cap \rho$.
From the proof of Lemma~\ref{lem:rho_xor_feasible_g} we observe that $(M(h) \setminus \{r\}) \cup \{r'\}$ is feasible for $h$,
therefore the edge $(r', h)$ blocks $M$ contradicting the stability of $M$.
\qed
%
\end{proof}

%% file: appendix_proof_struct_char.tex
\begin{proof}
\label{appendix:proof_struct_char}
We show that if $M$ is a feasible matching such that for any $M'$ the
set $\tilde{\mathcal{Y}}_{M \oplus M'}$ does not contain (1), (2), (3) as in Theorem~\ref{thm:charc},
then $M$ is popular in $G$.

Assume for the sake of contradiction that $M$ satisfies the conditions of Theorem~\ref{thm:charc},
and yet $M$ is not popular. Therefore there exists a feasible matching $M^*$ such that $\Delta (M^*, M) > 0$.
Consider the set ${\mathcal{Y}}_{M^* \oplus M}$. Recall that this set is a collection of paths
and cycles and the edges of $M^* \setminus M$ are labeled.
Let $\rho$ be any path or cycle in ${\mathcal{Y}}_{M^* \oplus M}$ and
let $\Delta(M^*, M)_{\rho}$ denote the difference between the votes of $M^*$ and $M$ when
restricted to the residents and hospitals in $\rho$. Since $\Delta (M^*, M) > 0$,
there exists a $\rho$ such that $\Delta(M^*, M)_{\rho} > 0$.
Note that $\rho$ is present in $\mathcal{Y}_{M^* \oplus M}$; using the presence of $\rho$ we establish
the existence of a $\rho' \in \tilde{\mathcal{Y}}_{M^* \oplus M}$ of the form (1), (2) or (3) which
contradicts our assumption.
We consider three cases depending on the structure of $\rho$.
\begin{enumerate}
\item {\bf $\rho$ is an alternating cycle or $\rho$ is an alternating path which starts and ends in an $M$ edge:} \\
Since $\rho \in  {\mathcal{Y}}_{M^* \oplus M}$,
and $\Delta(M^*, M)_\rho > 0$, it implies that there are more edges in $\rho$ labeled $(1, 1)$
than the number of edges labeled $(-1, -1)$.
We now delete the edges labeled $(-1, -1)$ from $\rho$; this breaks $\rho$ in to multiple alternating paths.
Note that each of these paths (say $\rho'$) start and end with an $M$ edge and are also present in $\tilde{\mathcal{Y}}_{M^* \oplus M}$.
Furthermore, since $\rho$ contained more number of edges labeled $(1,1)$ than the number of edges labeled
$(-1, -1)$, it is clear that there exists at least one $\rho'$ which has two edges labeled $(1,1)$. This
is a path of type (3) from the theorem statement and therefore contradicts our assumption that
$M$ satisfied the conditions of the theorem.

\item {\bf $\rho$ is an alternating path which starts or ends in an $M^*$ edge:}\\
The proof is similar to the previous case except that when we delete from $\rho$ the edges labeled $(-1, -1)$
we get paths $\rho' \in \tilde{\mathcal{Y}}_{M^* \oplus M}$ which are paths of type (2) or type (3) from the theorem
statement. This contradicts the assumption that $M$ satisfied the conditions of the theorem.
\end{enumerate}
This completes the proof of one direction of the Theorem. 
To prove the other direction, we prove the contrapositive of the statement.
That is, if for any feasible matching $M'$, $\tilde{\mathcal{Y}}_{M \oplus M'}$
contains (1), (2) or (3), then $M$ is not popular in $G$.
We first assume that $\rho \in \tilde{\mathcal{Y}}_{M \oplus M'}$ satisfying (1), (2), or (3) is also
present in ${\mathcal{Y}}_{M \oplus M'}$. Under this condition, it is possible to get a more popular matching
than $M$ by the following three cases.
\begin{itemize}
\item Let $M_2 = M \oplus \rho$ be a matching in $G$; by Lemma~\ref{lem:rho_xor_feasible_g}
we know that $M_2$ is
feasible in $G$. Comparing $M_2$ to $M$ yields two more votes for $M_2$.
Hence, $M_2$ is more popular than $M$.

\item If $\rho$ is an alternating path in $\tilde{\mathcal{Y}}_{M \oplus M'}$,
which has both its endpoints matched in $M$, and contains more than one edge labeled $(1,1)$.
Then similar to the case above $M_2 = M \oplus \rho$ is more popular than $M$.

\item If $\rho$ is an alternating path in $\tilde{\mathcal{Y}}_{M \oplus M'}$,
which has exactly one of its endpoints matched in $M$, and contains an edge labeled $(1,1)$,
then again $M_2 = M \oplus \rho$ is more popular than $M$.
\end{itemize}
Now let us assume that $\rho \in \tilde{\mathcal{Y}}_{M \oplus M'}$  is not present in ${\mathcal{Y}}_{M \oplus M'}$.
In such a case, $\rho$ is contained in a larger path or a cycle $\rho' \in {\mathcal{Y}}_{M \oplus M'}$ obtained by combining $\rho$ with
other paths in $\tilde{\mathcal{Y}}_{M \oplus M'}$
and adding the deleted $(-1, -1)$ edges. Using the larger path or cycle $\rho'$ we can construct a matching that is more
popular than $M$. Note that we need to use paths or cycles in ${\mathcal{Y}}_{M \oplus M'}$ to obtain another matching, since  we have to ensure that the matching obtained is indeed feasible
in the instance and the correspondences are maintained.
\qed
\end{proof}

%% file: max-card-pop-trimmed.tex
\section{Popular matchings in LCSM$^+$ problem }
\label{sec:max-card-pop-lcsm}
In this section we present efficient algorithms for computing
(i) a maximum cardinality popular matching, and
(ii) a matching that is popular amongst all the maximum cardinality matchings
in a given LCSM$^+$ instance.
Our algorithms are inspired by the reductions of Kavitha and Cseh~\cite{CsehK16}
where they work with a stable marriage instance.
We describe a general reduction from an LCSM$^+$ instance $G$ to another
LCSM$^+$ instance $G_s$. Here $s = 2, \ldots, |\RR|$.
The algorithms for the two problems are obtained by choosing an appropriate value of $s$.

\noindent {\bf The graph $G_s$:}
Let $G = (\RR \cup \HH, E)$ be the input LCSM$^+$ instance.
The graph $G_s = (\RR_s \cup \HH_s, E_s)$ is constructed as follows:
Corresponding to every resident $r \in \RR$, we have $s$ copies of $r$,
call them $r^0, \ldots , r^{s-1}$ in $\RR_s$.
The hospitals in $\HH$ and their capacities remain unchanged; however
we have additional dummy hospitals each of capacity $1$.
Corresponding to every resident $r \in \RR$, we have $(s - 1)$ dummy hospitals
$d_{r}^0, \ldots, d_{r}^{s-2}$ in $\HH_s$.
Thus,
\begin{eqnarray*}
\RR_s = \{\ r^0, \ldots , r^{s-1} \  \ |  \ \forall r \in \RR \} ; \ \
\HH_s = \HH \cup \{\ d_{r}^0, \ldots, d_{r}^{s-2} \ \ | \ \forall r \in \RR \}
\end{eqnarray*}
We use the term level-$i$ resident for a resident $r^i \in \RR_s$ for $0 \le i \le s-1$.
The preference lists corresponding to $s$ different residents of $r$ in $G_s$
are:
\begin{itemize}
\item For a level-$0$ resident $r^0$, its preference list in $G_s$
is the preference list of $r$ in $G$, followed by the dummy hospital $d_{r}^0$.

\item For a level-$i$ resident $r^i$, where $ 1 \le i \le s-2$,
its preference list  in $G_s$ is $d_{r}^{i-1}$ followed by preference list
of $r$ in $G$, followed by $d_{r}^i$.

\item For a level-$(s-1)$ resident $r^{s-1}$, its preference list in $G_s$
is the dummy hospital $d_{r}^{s-2}$ followed by the preference list of $r$ in $G$.
\end{itemize}
The preference lists of hospitals in $G_s$ are as follows.
\begin{itemize}
\item The preference list for a dummy hospital $d_{r}^i$ is $r^i$
followed by $r^{i+1}$. 
\item For $h \in \HH$, its preference list in $G_s$, has level-$(s-1)$ residents
followed by level-$(s-2)$ residents, so on upto the level-$0$ residents
in the same order as in $h$'s preference list in $G$.
\end{itemize}

\noindent Finally, we need to specify the classifications of the hospitals in $G_s$.
For every class $C_i^h$ in the instance $G$, we have a corresponding class
$\bar{C}_i^h = \bigcup_{r \in C_i^h}{\{r^0, \ldots, r^{s-1}\}}$ in $G_s$, such
that $q(\bar{C}_i^h) = q(C_i^h)$.
We note that $|\bar{C}_i^h| = s \cdot |C_i^h|$.
Let $M_s$ be a stable matching in $G_s$.
Then $M_s$ satisfies the following properties:
\begin{itemize}
\item [($\mathcal{I}_1$)] Each $d_r^i \in \HH_s$ for $0 \le i \le s-2$,
is matched to one of $\{r^i, r^{i+1}\}$ in $M_s$.

\item [($\mathcal{I}_2$)] The above invariant implies that for every $r \in \RR$
at most one of $\{r^0, \ldots, r^{s-1}\}$ is assigned to a non-dummy hospital in $M_s$.

\item [($\mathcal{I}_3$)] For a resident $r \in \RR$, if $r^i$ is
matched to a non-dummy hospital in $M_s$, then
for all $  0 \le j \le i-1$, $M_s(r^j) = d_r^j$.
Furthermore,
for all $ i+1 \le p \le s-1 $, $M_s(r^p) = d_r^{p-1}$.
This also implies that in $M_s$ all residents $r^0, \ldots, r^{s-2}$  are matched
and only $r^{s-1}$ can be left unmatched in $M_s$.
\end{itemize}

These invariants allow us to naturally map the stable matching $M_s$
to a feasible matching $M$ in $G$. We define a function $map(M_s)$ as follows.
\begin{eqnarray*}
M = map(M_s) = \{(r, h) : h \in \HH  \mbox { and } (r^i, h) \in M_s \mbox { for exactly one of } 0 \le i \le s-1 \}
\end{eqnarray*}

\noindent We outline an algorithm that computes a feasible
matching in an LCSM$^+$ instance $G$.
Given $G$ and $s$, construct the graph $G_s$ from $G$.
Compute a stable matching $M_s$ in $G_s$. If $G$ is an LCSM$^+$ instance we use
the algorithm of Huang~\cite{Huang2010} to compute a stable matching in $G$.
If $G$ is a PCSM$^+$ instance, it is easy to observe that $G_s$ is also a PCSM$^+$ instance.
In that case, we use the algorithm of Abraham~et~al.\cite{Abraham2007}
to compute a stable matching.
(The SPA instance is different from a PCSM$^+$ instance, however, there is a easy reduction from the PCSM$^+$ instance to SPA,
we give the reduction (refer Appendix~\ref{appendix:reduction_pcsm_spa}) for the sake of completeness.). 
We output $M = map(M_s)$ whose feasibility is guaranteed by the invariants mentioned earlier.
The complexity of our algorithm depends on $s$ and the
time required to compute a stable matching in the problem instance.

%

In the rest of the paper, we denote by $M$ the matching obtained
as $map(M_s)$ where $M_s$ is a stable matching in $G_s$.
For any resident $r_i \in \RR$, we define
 \begin{eqnarray*}
map^{-1}(r_i, M_s) &=&  r_i^{j_i} \ \ \ \ \  \mbox{ where $0 \le j_i \le s-1$ and $M_s(r_i^{j_i})$ is a non-dummy hospital}\\
 &=& r_i^{s-1}   \ \ \ \mbox { otherwise.}
\end{eqnarray*}
Recall by Invariant ($\mathcal{I}_3$), exactly one of the level copy of $r_i$ in $G_s$ is matched to a non-dummy hospital in $M_s$.
For any feasible matching $M'$ in $G$ consider the set $\mathcal{Y}_{M \oplus M'}$ --
recall that this is a collection of $M$ alternating paths and cycles in $G$.
For any path or cycle $\rho$ in $\mathcal{Y}_{M \oplus M'}$, let us denote
by $\rho_{s} = map^{-1}(\rho, M_s)$ the path or cycle in $G_s$ obtained by
replacing every resident $r$ in $\rho$ by $map^{-1}(r, M_s)$.
Recall that if a resident $r$ is present in the class $C^h_j$
defined by a hospital $h$ in $G$, then in the graph $G_s$, $r^i \in \bar{C}^h_j$
for $i = 0, \ldots, s-1$.
The $map^{-1}$ function maps a resident $r$ in $G$ to a
unique level-$i$ copy in $G_s$.
Using Lemma~\ref{lem:rho_xor_feasible_g} and these observations we get the
following corollary.
\begin{corollary}
\label{lem:rho_xor_feasible_gs}
Let $\rho$ be an alternating path or an alternating cycle in
$\mathcal{Y}_{M \oplus M'}$, then $M_s \oplus \rho_s$ is a feasible matching
in $G_s$, where $\rho_s = map^{-1}(\rho, M_s)$.
\end{corollary}
The following technical lemma is useful in proving the properties of the matchings produced by our algorithms.
\begin{lemma}
\label{lem:blocking-pair-ms}
Let $\rho$ be an alternating path or an alternating cycle in
$\mathcal{Y}_{M \oplus M'}$, and $\rho_s = map^{-1}(\rho, M_s)$.
\begin{enumerate}
\item There cannot be any edge labeled $(1, 1)$ in $\rho_s$.
\item Let $\langle r_a^{j_a}, h, r_b^{j_b} \rangle$ be a sub-path of
$\rho_s$, where $h = M_s(r_b^{j_b})$.
Then, the edge $(r_a^{j'_a}, h) \notin \rho_s$ cannot be labeled $(1, 1)$,
where $j'_a < j_a$.
\end{enumerate}
\end{lemma}
\input{lemma2-proof}

\subsection{Maximum cardinality popular matching} 
\label{sec:sec3}
Let $G = (\RR \cup \HH, E)$ be an instance of the LCSM$^+$ problem where we are
interested in computing a maximum cardinality popular matching.
We use our generic reduction with the value of the parameter $s = 2$.
Since $G_2$ is linear in the size of $G$,
and a stable matching in an LCSM$^+$ instance can be computed in $O(mn)$ time~\cite{Huang2010},
we obtain an $O(mn)$ time algorithm to compute a maximum cardinality popular matching in $G$.
In case $G$ is a PCSM$^+$ instance, we use the linear time algorithm in
\cite{Abraham2007} for computing a stable matching to get a linear time algorithm
for our problem.
The proof of correctness involves two things -- we first show that
$M$ is popular in $G$. We then argue that it is the largest size popular matching in $G$.
We state the main theorem of this section below.
\begin{theorem}
\label{thm:pop1}
Let $M = map(M_2)$ where $M_2$ is a stable matching in $G_2$. Then $M$ is a maximum
cardinality popular matching in $G$.
\end{theorem}

We break down the proof of Theorem~\ref{thm:pop1} in two parts.
Lemma~\ref{lem:pop-in-g} shows that the assignment $M$ satisfies all the
conditions of Theorem~\ref{thm:charc}.
Lemma~\ref{lem:largest-pop} shows that the matching
output is indeed the largest size popular matching in the instance.
Let $M'$ be any assignment in $G$.
Recall the definition of $\tilde{\mathcal{Y}}_{M \oplus M'}$ -- this set contains $M$ alternating
paths and $M$ alternating cycles in $G$ and the edge labels on the $M'$ edges belong to $\{(-1, 1), (1, -1), (1, 1)\}$.

\begin{lemma}
\label{lem:pop-in-g}
Let $M = map(M_2)$ where $M_2$ is a stable matching in $G_2$ and let $M'$ be any feasible assignment in $G$. Consider
the set of alternating paths and alternating cycles $\tilde{\mathcal{Y}}_{M \oplus M'}$.
Then, the following hold:
\begin{enumerate}
\item An alternating cycle $C$ in $\tilde{\mathcal{Y}}_{M \oplus M'}$,
does not contain any edge labeled $(1, 1)$.
\item An alternating path $P$ in $\tilde{\mathcal{Y}}_{M \oplus M'}$
that starts or ends with an edge in $M'$, does not contain any edge labeled $(1, 1)$.
\item An alternating path $P$ in $\tilde{\mathcal{Y}}_{M \oplus M'}$ which
starts and ends with an edge in $M$, contains at most one edge labeled $(1, 1)$.

\end{enumerate}
\end{lemma}
\input{lemma3-proof.tex}

\begin{lemma}
\label{lem:no-aug-path}
There is no augmenting path with respect to $M$ in $\tilde{\mathcal{Y}}_{M \oplus M'}$.
\end{lemma}
\begin{proof}
Let $P = \langle r_0, h_1, r_1, h_2, \ldots, h_{k-1}, r_{k-1}, h_k \rangle$ be an augmenting
path where for each $i=1, \ldots, k-1$, $M(r_i) = h_i$.
The existence of $P$ in $\tilde{\mathcal{Y}}_{M \oplus M'}$
implies that there exists an $M_2$ augmenting path
$P_2 = \langle r_0^{j_0}, h_1, r_1^{j_1}, h_2, \ldots, h_{k-1}, r_{k-1}^{j_{k-1}}, h_k \rangle$
in $G_2$, and for $t = 0, \ldots, k-1$, $j_t \in \{0, 1\}$.

Using invariants ($\mathcal{I}_1$), ($\mathcal{I}_2$), and ($\mathcal{I}_3$),
we conclude that a resident $r$ remains unmatched in $M_2$ when its level-$0$ copy is
matched to the dummy vertex $d_r$, and the level-$1$ copy is unmatched in $M_2$.
Therefore the first resident on the path $P_2$ is a level-$1$ resident.
The second resident on the path $r_1$ has to be a level-$1$ resident,
otherwise the edge $(r_0^1, h_1)$ will be labeled $(1, 1)$, and thus
contradict the stability of $M_2$ (using Lemma~\ref{lem:blocking-pair-ms}(1)).
This is because $r_0^1$ prefers being matched to $h_1$ than being unmatched in $M_2$,
and $h_1$ prefers level-$1$ resident over a level-$0$ resident.
Observe that $j_{k-1} = 0$, else the pair $(r_{k-1}^0, h_k)$ is labeled $(1, 1)$,
as $r_{k-1}^0$ is matched to $d_r$ (by invariants ($\mathcal{I}_1$) and ($\mathcal{I}_2$)),
which is at the end of its preference list, and $h_k$ is unmatched in $M'$.

Therefore the path $P_2$ is of the form
$\langle r_0^1, h_1, r_1^1, h_2, \ldots, h_{k-1}, r_{k-1}^0, h_k \rangle$.
As $j_0 = j_1 = 1$ and $j_{k-1} = 0$,
there exists an index $x$ in $P_2$ such that there is a transition from a
level-1 resident to a level-0 resident.
That is, $(r_x^0, h_x) \in M_2$ and $(r_{x-1}^1, h_x) \notin M_2$ both belong to $P_2$.

We enumerate the possible labels for the edge $e_x =(r_{x-1}, h_x)$ in $G$.
\begin{itemize}
\item If $e_x$ is labeled $(1,1)$ or $(1,-1)$, then the edge
    $(r_{x-1}^1, h_x)$ is labeled $(1, 1)$, which contradicts
    the stability of $M_s$ (using Lemma~\ref{lem:blocking-pair-ms}(1)).

\item If $e_x$ is labeled $(-1,1)$, then the edge $(r_{x-1}^0, h_x)$
    is labeled $(1, 1)$, which contradicts
    the stability of $M_s$ (using Lemma~\ref{lem:blocking-pair-ms}(2)).
\end{itemize}

This contradicts our assumption that $P$ is augmenting with
respect to $M$ in $\tilde{\mathcal{Y}}_{M \oplus M'}$.
\qed
\end{proof}

\begin{lemma}
\label{lem:largest-pop}
There exists no popular matching $M^*$ in $G$ such that $|M^*| > |M|$.
\end{lemma}
\begin{proof}
For contradiction, assume that such an assignment $M^*$ exists in $G$.
Consider the set ${\mathcal{Y}}_{M \oplus M^*}$; recall
that this set contains alternating paths and cycles possibly containing
edges
labeled $(-1, -1)$. Since $|M^*| > |M|$ there must exist
an augmenting path $P$ in ${\mathcal{Y}}_{M \oplus M^*}$. We first claim
that the path $P$ must contain at least one edge labeled $(-1, -1)$. If not,
then the path $P$ is also contained in $\tilde{\mathcal{Y}}_{M \oplus M^*}$.
However, by Lemma~\ref{lem:no-aug-path} there is no augmenting path with respect
to $M$ in $\tilde{\mathcal{Y}}_{M \oplus M'}$ for
any feasible matching $M'$ in $G$.

We now remove all edges from $P$ which are labeled $(-1, -1)$. This
breaks the path into sub-paths say $P_1, P_2, \ldots, P_t$ for some $t\ge 1$,
where $P_1$ and $P_t$ have one endpoint unmatched in $M$.
Consider the path $P_1$; since $P_1$ does not contain any $(-1,-1)$
edge this implies that $P_1 \in \tilde{\mathcal{Y}}_{M \oplus M^*}$.
Without loss of generality, assume that $P_1$ starts with a resident
$r$ which is unmatched in $M$. Thus using Lemma~\ref{lem:pop-in-g}(2),
$P_1$ does not contain any edge labeled $(1, 1)$.
Let us denote by $\Delta(M^*, M)_{P_1}$ the difference between votes
of $M^*$ and $M$ restricted to vertices of path $P_1$.
It is clear that $\Delta(M^*, M)_{P_1} < 0$.
Also, for each $i = 2, \ldots, t-1$, the alternating paths $P_i$ have both
of their endpoints matched in $M$. Thus we have $\Delta(M^*, M)_{P_i} \le 0$ as there
can be at most one $(1,1)$ edge (by Lemma~\ref{lem:pop-in-g}(3)) in these paths,
but the endpoints prefer $M$, as they are matched in $M$ but not in $M'$.
If $P_t$ exists, then a argument similar as given for $P_1$, we have $\Delta(M^*, M)_{P_t} < 0$.
Using these observations, we conclude that $M$ is more popular than $M^*$,
a contradiction to the assumption that $M^*$ and $M$ are both popular.

Thus, for any given matching $M^*$ such that $|M^*| > |M|$, we know
that $M$ is more popular than such a matching.
This completes the proof of the lemma, and shows that the matching $M = map(M_2)$
is a maximum cardinality popular matching in $G$.
\qed
\end{proof}

%% file: lemma2-proof.tex
\begin{proof}
\label{lem2:proof}
Let $\langle r_a^{j_a}, h, r_b^{j_b} \rangle$ be a sub-path of
$\rho_s$, where $h = M_s(r_b^{j_b})$
(Figure~\ref{fig:ctree-alt-path}).
\input{app-fig-sec4}
As $M_s \oplus \rho_s$ is feasible in $G_s$ (Corollary~\ref{lem:rho_xor_feasible_gs}),
the set $(M_s(h) \setminus \{ r_b^{j_b} \}) \cup \{ r_a^{j_a} \}$ is feasible
for $h$ in $G_s$. Now since $(r_a^{j_a}, h)$ is labeled $(1, 1)$, the edge
$(r_a^{j_a}, h)$ blocks $M_s$
contradicting its stability. This proves (1). To prove (2),
assume that the edge $(r_a^{j'_a}, h) \notin \rho_s$ is labeled $(1, 1)$.
The residents $r_a^{j_a}$ and $r_a^{j'_a}$ belong to the same class (say $\bar{C}^h_k$)
in $G_s$, hence
$(M_s(h) \setminus \{ r_b^{j_b} \}) \cup \{ r_a^{j'_a} \}$
is  feasible for $h$.
Thus the edge $(r_a^{j'_a}, h)$ blocks $M_s$
contradicting its stability.
\end{proof}

%% file: app-fig-sec4.tex
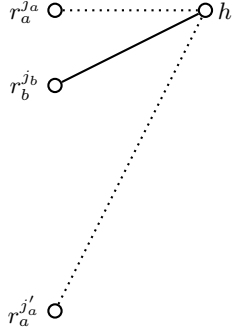
\begin{figure}[ht]
\centering
\begin{tikzpicture}[every node/.style={circle, draw, thick, inner sep=0pt, minimum width=5pt}]
\node (r_aa)[label=left:$r_a^{j_a}$] at (0, 4) {};
\node (r_bb)[label=left:$r_b^{j_b}$] at (0, 3) {};
\node (r_ac)[label=left:$r_a^{j'_a}$] at (0, 0) {};
\node (h)[label=right:$h$] at (2, 4) {};

\draw[thick, dotted] (r_aa) -- (h);
\draw[thick]         (r_bb) -- (h);
\draw[thick, dotted] (r_ac) -- (h);
\end{tikzpicture}
\caption{The edges $(r_a^{j_a}, h)$ and $(r_b^{j_b}, h)$ belong to $\rho$,
while the edge $(r_a^{j'_a}, h)$ does not belong to $\rho$.}
\label{fig:ctree-alt-path}
\end{figure}

%% file: lemma3-proof.tex
\begin{proof}
We first prove the parts (1) and (2).
Recall that $M = map(M_2)$ where $M_2$ is a stable matching in $G_2$.
Assume that $\rho = \langle u_0, v_1, u_1, \ldots, v_k, u_k \rangle$ where
for each $i=0, \ldots, k$, $v_i = M(u_i)$ (in case $u_i$ is a hospital, $v_i \in M(u_i)$).
In case $\rho$ is a cycle, all subscripts follow mod $k$ arithmetic.
The existence of $\rho$ in $\tilde{\mathcal{Y}}_{M \oplus M'}$ implies that there is
an associated $M_2$ alternating path or an $M_2$ alternating cycle
$\rho_2 = map^{-1}(\rho, M_2)$ in $G_2$.

Now assume for the sake of contradiction that $\rho$ contains an edge
$e = (r_a, h_b) \notin M$ labeled $(1, 1)$ for some $a = 0, \ldots, k$,
and $b = 0, \ldots, k$.
We observe the following about preferences of $r_a$ and $h_b$ in $G$.

\begin{itemize}
\item [($\mathcal{O}_1$)] $r_a$ prefers $h_b$ over $h_a = M(r_a)$.
\item [($\mathcal{O}_2$)] $h_b$ prefers $r_a$ over $r_b \in M(h_b)$, where $r_b = {\bf corr}(r_a)$.
\end{itemize}

Using the presence of an edge labeled $(1, 1)$ in
$\rho$, we will contradict the stability of $M_2$ in $G_2$.
Consider the edge $e' = (r_a^{j_a}, h_b)$ in $G_2$.
Since $h_b = M_2(r_b^{j_b})$, we observe that $e' \notin M_2$.
We consider the four cases that can arise depending on the values of $j_a$ and $j_b$.

\vspace{0.1in}
\begin{enumerate}
\begin{minipage}{0.6\linewidth}
\item $j_a = j_b = 0$
\item $j_a = j_b = 1$
\end{minipage}
\begin{minipage}{0.45\linewidth}
\item $j_a = 1$ and $j_b = 0$
\item $j_a = 0$ and $j_b = 1$
\end{minipage}
\end{enumerate}

Recall observation $(\mathcal{O}_1)$, and the fact that the residents do not change their
preferences in $G_2$ w.r.t. the hospitals originally in $G$.
This implies in all the four cases above, the resident $r_a^{j_a}$
prefers $h_b$ over $h_a = M_2(r_a^{j_a})$.
Using $(\mathcal{O}_2)$ and the fact that a hospital $h$ in $G_2$ prefers level-$1$
residents over level-$0$ residents, we can conclude the following.
For the cases $(1), (2)$ and $(3)$, hospital $h_b$ prefers $r_a^{j_a}$ over $r_b^{j_b}$,
which implies that the pair $(r_a^{j_a}, h_b)$ is labeled $(1, 1)$, and thus
forms a blocking pair w.r.t. $M_2$ (using Lemma~\ref{lem:blocking-pair-ms}(1)).

We now consider the three different cases for $\rho$ depending on whether $\rho$
is a path or a cycle.
When $\rho$ is a path, we break down its proof in two cases --
(i) $\rho$ starts or ends with a resident unmatched in $M$.
(ii) $\rho$ starts or ends with an under-subscribed hospital.
In each of the different possibilities for $\rho$, we show that the stability
of $M_2$ can be contradicted even in case $(4)$, i.e. when $j_a = 0$ and $j_b = 1$.

\begin{itemize}
\item $\rho = \langle r_0, h_1, r_1, \ldots, h_{k-1}, r_{k-1} \rangle$ is an
    alternating path that starts or ends with a resident which is unmatched in $M$.
    Here $\rho_2 = map^{-1}(\rho, M_2) = \langle r_0^{j_0}, h_1, r_1^{j_1}, \ldots, h_{k-1}, r_{k-1}^{j_{k-1}} \rangle$
    and for $t = 0, \ldots, k-1$, $j_t \in \{0, 1\}$.

    Using invariants ($\mathcal{I}_1$), ($\mathcal{I}_2$), and ($\mathcal{I}_3$),
    we conclude that a resident $r$ remains unmatched in $M_2$ when its level-$0$ copy is
    matched to the dummy hospital $d_r$, and the level-$1$ copy is unmatched in $M_2$.
    Therefore, the first resident on the path $\rho_2$ is a level-$1$ resident.
    Furthermore, the second resident on the path $r_1$ has to be a level-$1$ resident.
    Otherwise, as $r_0^1$ is unmatched in $M_2$ and $h_1$ prefers a level-$1$ resident
    over a level-$0$ resident, the edge $(r_0^1, h_1)$ will be labeled $(1, 1)$,
    and thus forms a blocking pair w.r.t. $M_2$ (using Lemma~\ref{lem:blocking-pair-ms}(1)).

    We consider an edge $e \in \rho$ such that $b = a+1$.
    In case $(4)$, we observe that as $j_0 = j_1 = 1, j_a = 0$, and $a < b$,
    there exists an index $x$ in $\rho_2$ such that there is a transition from
    a level-$1$ resident to a level-$0$ resident.
    That is, $(r_x^0, h_x) \in M_2$ and $(r_{x-1}^1, h_x) \notin M_2$ both belong to $\rho_2$.

    We enumerate the possible labels for the edge $e_x =(r_{x-1}, h_x)$ in $G$.
    \begin{itemize}
    \item If $e_x$ is labeled $(1,1)$ or $(1,-1)$, then the edge
        $(r_{x-1}^1, h_x)$ is labeled $(1, 1)$, and thus blocks $M_2$
        (using Lemma~\ref{lem:blocking-pair-ms}(1)).

    \item If $e_x$ is labeled $(-1,1)$, then the edge $(r_{x-1}^0, h_x)$
        is labeled $(1, 1)$, and thus blocks $M_2$
        (using Lemma~\ref{lem:blocking-pair-ms}(2)).
    \end{itemize}

\item $\rho = \langle h_0, r_1, h_1, \ldots, r_{k-1}, h_{k-1} \rangle$ is an alternating
    path that starts or ends with an under-subscribed hospital.
    Here $\rho_2 = map^{-1}(\rho, M_2) = \langle h_0, r_1^{j_1}, h_1, \ldots, r_{k-1}^{j_{k-1}}, h_{k-1} \rangle$
    and for $t = 1, \ldots, k-1$, $j_t \in \{0, 1\}$.

    Observe that if $j_1 = 1$, then $(r_1^0, h_0)$ is labeled $(1, 1)$,
    as $h_0$ is unmatched in $M$, and $r_1^0$ prefers $h_0$ to $d_{r_1}$
    ($d_{r_1} = M_2(r_1^0)$ using invariants ($\mathcal{I}_1$), ($\mathcal{I}_2$),
    and ($\mathcal{I}_3$)), contradicting the stability of $M_s$
    (using Lemma~\ref{lem:blocking-pair-ms}(1)).
    Thus, it must be the case that $j_0 = 0$.
    Note that the edge $(r_1, h_0)$ can not be labeled $(1,1)$ in $G$,
    as $h_0$ being under-subscribed prefers being matched to $r_1$,
    and residents do not change their votes, and thus the edge $(r_1^0, h_0)$
    is labeled $(1, 1)$, contradicting the stability of $M_2$
    (using Lemma~\ref{lem:blocking-pair-ms}(1)).

    We consider an edge $e \in \rho$ such that $a = b+1$.
    In case $(4)$, we observe that as $j_1 = 0 , j_b = 1$, and $a > b$, there exists
    an index $x$ in $\rho_2$ such that there is a transition from a
    level-$0$ resident to a level-$1$ resident.
    That is, $(r_x^0, h_x) \in M_2$ and $(r_{x+1}^1, h_x) \notin M_2$ both belong to $\rho_2$.
    Using an argument similar to in the case above, we can show that
    either the edge $(r_{x+1}^1, h_x)$ or the edge $(r_{x+1}^0, h_x)$
    is labeled $(1, 1)$, and therefore forms a blocking pair w.r.t. $M_2$.

\item $\rho = \langle r_0, h_0, r_1, h_1, \ldots, r_k, h_k, r_0 \rangle$ is an alternating cycle.
    Here $\rho_2 = map^{-1}(\rho, M_s) = \langle r_0^{j_0}, h_0, r_1^{j_1}, h_1, \dots, r_k^{j_k}, h_k, r_0^{j_0} \rangle$
    and for $t = 1, \ldots, k-1$, $j_t \in \{0, 1\}$.

    We consider an edge $e \in \rho$ such that $a = b+1$.
    As $j_a = 0$ and $j_b = 1$, and $b < a$, this is a transition from a
    level-$1$ resident to a level-$0$ resident in the cycle $\rho_2$.
    To complete the cycle $\rho_2$ there must exist an index $x$ such that there
    is a transition from a level-$0$ resident to a level-$1$ resident.
    That is, $(r_x^0, h_x) \in M_2$ and $(r_{x+1}^1, h_x) \notin M_2$ both belong to $\rho_2$.
    Using an argument similar to as in the first case, we can show that
    either the edge $(r_{x+1}^1, h_x)$ or the edge $(r_{x+1}^0, h_x)$
    is labeled $(1, 1)$, and therefore forms a blocking pair w.r.t. $M_2$.
\end{itemize}

We now prove part (3) of the lemma.
Consider $P = \langle r_0, h_0, \ldots, r_{k-1}, h_{k-1} \rangle$ where
for each $i=0, \ldots, k-1$, $M(r_i) = h_i$.
The existence of $P$ in $\tilde{\mathcal{Y}}_{M \oplus M'}$
implies that there exists an $M_2$ alternating path $P_2$ in $G_2$.
Here $P_2 = \langle r_0^{j_0}, h_0, \ldots, r_{k-1}^{j_{k-1}}, h_{k-1} \rangle$,
and for $t = 0, \ldots, k-1$, $j_t \in \{0, 1\}$.

For the sake of contradiction assume that $P$ contains
at least two edges, $e_1 = (r_{x}, h_{x -1})$, $e_2 = (r_{y}, h_{y -1})$
for some $x, y = 1, \ldots, k-1$, w.l.o.g. $x \neq y, x < y$ and $e_1, e_2$ are labeled $(1, 1)$.
We observe the following about preferences of $r_{x}$, $r_{y}$
and $h_{x -1}$, $h_{y -1}$ in $G$.

\begin{itemize}
\item [($\mathcal{O}_1$)] $r_{x}$ prefers $h_{x -1}$ over $h_{x} = M(r_{x})$.\\
$r_{y}$ prefers $h_{y -1}$ over $h_{y} = M(r_{y})$.

\item [($\mathcal{O}_2$)] $h_{x -1}$ prefers $r_{x}$ over $r_{x -1} \in M(h_{x -1})$.\\
$h_{y -1}$ prefers $r_{y}$ over $r_{y -1} \in M(h_{y -1})$.
\end{itemize}

Using the presence of the edges $e_1$ and $e_2$ labeled $(1, 1)$ in
$P$, we will contradict the stability of $M_2$ in $G_2$.
Consider the edges $e'_1 = (r_x^{j_x}, h_{x -1})$ and
$e'_2 = (r_y^{j_y}, h_{y -1})$ in $G_2$,
and since $h_{x-1} = M_2(r_x^{j_{x-1}})$ and $h_{y-1} = M_2(r_y^{j_{y-1}})$,
note that $e'_1, e'_2 \notin M_2$.

We first consider the edge $e'_1$, and consider the four cases that
can arise depending on the values of $j_x$ and $j_{x -1}$.

\vspace{0.1in}
\begin{enumerate}
\begin{minipage}{0.6\linewidth}
\item $j_{x -1} = j_x = 0$
\item $j_{x -1} = j_x = 1$
\end{minipage}
\begin{minipage}{0.45\linewidth}
\item $j_{x -1} = 0$ and $j_x = 1$
\item $j_{x -1} = 1$ and $j_x = 0$
\end{minipage}
\end{enumerate}

Recall observation $(\mathcal{O}_1)$, and the fact that the residents do not change their
preferences in $G_2$ w.r.t the hospitals originally in $G$.
This implies that in all the four cases above, the resident $r_x^{j_x}$
prefers $h_{x -1}$ over $h_x = M_2(r_x^{j_x})$.
Using $(\mathcal{O}_2)$ and the fact that a hospital $h$ in $G_2$ prefers level-$1$
residents over level-$0$ residents, we can conclude the following.
For the cases $(1), (2)$ and $(3)$, hospital $h_{x -1}$ prefers $r_x^{j_x}$ over $r_{x-1}^{j_{x -1}}$,
which implies that the pair $(r_x^{j_x}, h_{x -1})$ is labeled $(1, 1)$,
which contradicts the stability of $M_s$ (using Lemma~\ref{lem:blocking-pair-ms}(1)).

With a similar analysis for the edge $e'_2$, we conclude that the first three cases do not arise.
There is only one case left to consider, when $j_{x -1} = 1, j_x = 0$ and $j_{y -1} = 1, j_y = 0$.
As $x \neq y, x < y$, and $j_x = 0$, $j_{y-1} = 1$, there exists an index $\ell$ in $P_2$
such that there is a transition from a level-0 resident to a level-1 resident.
That is, $(r_\ell^0, h_\ell) \in M_2$ and $(r_{\ell +1}^1, h_\ell) \notin M_2$ both belong to $P_2$.

We enumerate the possible labels for the edge $e_\ell =(r_{\ell +1}, h_\ell)$ in $G$.
\begin{itemize}
\item If $e_\ell$ is labeled $(1,1)$ or $(1,-1)$, then the edge
    $(r_{\ell +1}^1, h_\ell)$ is labeled $(1, 1)$, which contradicts
    the stability of $M_s$ (using Lemma~\ref{lem:blocking-pair-ms}(1)).

\item If $e_\ell$ is labeled $(-1,1)$, then the edge
    $(r_{\ell +1}^0, h_\ell)$ is labeled $(1, 1)$, which contradicts
    the stability of $M_s$ (using Lemma~\ref{lem:blocking-pair-ms}(2)).
\end{itemize}

This completes the proof.
\qed
\end{proof}

%% file: pop-amongst-max.tex
\subsection{Popular matching amongst maximum cardinality matchings}
\label{sec:sec4}
In this section we give an efficient algorithm for computing a matching which is
{\it popular} amongst the set of maximum cardinality matchings.
The matching $M$ that we output cannot be beaten in terms of
votes by any feasible maximum cardinality matching.
Our algorithm uses the generic reduction with a value of $s = |\RR|$ = $n_1$ (say).
Thus, $|\RR_{n_1}| = n_1^2$, and $|\HH_{n_1}| = |\HH| + O(n_1^2)$.
Furthermore,  $|E_{n_1}| = O(mn_1)$ where
$m = |E|$. Thus the running time of the generic algorithm presented earlier with $s = n_1$
for an LCSM$^+$ instance is $O(mn\cdot n_1) = O(mn^2)$
and for a PCSM$^+$ instance is $O(mn_1) = O(mn)$.

To prove correctness, we show that the 
matching output by our algorithm is
(i) maximum cardinality and
(ii) popular amongst all maximum cardinality feasible matchings.
Let $M = map(M_{n_1})$ and $M^*$ be any maximum cardinality feasible matching in $G$.
Consider the set $\mathcal{Y}_{M \oplus M^*}$, and let $\rho$ be an
alternating path or an alternating cycle in $\mathcal{Y}_{M \oplus M^*}$.
Let $\rho_{n_1} = map^{-1}(\rho, M_{n_1})$ denote the associated alternating path or cycle in $G_{n_1}$.
We observe that every hospital on the path $\rho_{n_1}$ is a non-dummy hospital
since $\rho_{n_1}$ was obtained using the inverse-map of $\rho$.
We observe  two useful properties about such a path or cycle $\rho_{n_1}$ in $G_{n_1}$.
We show that if for a hospital $h \in \rho_{n_1}$, the level of the unmatched resident incident on 
$h$ is greater than the level of the matched resident incident on $h$, then
such a level change is {\it gradual}, and the associated edge in $\rho$ has
the label $(-1, -1)$. 
Lemma~\ref{lem:no-jump-transition}, gives a proof of these.

\input{lemma-no-jump}

We use Lemma~\ref{lem:no-aug-path-maxc} to prove that $M$ is a maximum cardinality matching in $G$.
\begin{lemma}
\label{lem:no-aug-path-maxc}
Let $M^{*}$ be any feasible maximum cardinality matching in $G$.
Then there is no augmenting path with respect to $M$ in $\mathcal{Y}_{M \oplus M^*}$.
\end{lemma}
\begin{proof}
For the sake of contradiction assume that the path $P = \langle r_0, h_1, r_1, \ldots, h_{k-1}, r_{k-1}, h_k \rangle$
is an augmenting path where for each $i = 1,\ldots, (k-1)$, $M(r_i) = h_i$.
Here $r_0$ is unmatched in $M$, and $h_k$ is under-subscribed in $M$.
The existence of $P$ in $\mathcal{Y}_{M \oplus M^*}$
implies that there exists an $M_{n_1}$ augmenting path
$P_{n_1} = map^{-1}(P, M_{n_1}) = \langle r_0^{j_0}, h_1, r_1^{j_1}, \ldots, h_{k-1}, r_{k-1}^{j_{k-1}}, h_k \rangle$
in $G_{n_1}$, and for $t = 0, \ldots, k-1$, $j_t \in \{0, \ldots , n_1-1\}$,
where $h_i = M_{n_1}(r_i^{j_i})$.

Since $r_0$ is unmatched in $M$, by invariant ($\mathcal{I}_3$), it implies
that for $0 \le i \le n_1-2$, $M_{n_1}(r_0^i) = d_{r_0}^i$, and $r_0^{n_1-1}$ is unmatched in $M_{n_1}$.
This implies that the first resident in the path $P_{n_1}$ is a level-$(n_1-1)$ resident $r_0^{n_1-1}$.
The second resident on the path $P_{n_1}$ also has to be a level-$(n_1-1)$ resident.
If not, then the edge $(r_0^{n_1-1}, h_1)$ is labeled $(1, 1)$ since $h_1$ prefers
$r_0^{n_1-1}$ to any resident at a level lower than $n_1-1$ and $r_0^{n_1-1}$ is unmatched in $M_{n_1}$.
The last resident in the path $P_{n_1}$ is a level-$0$ resident i.e. $r_{k-1}^{j_{k-1}} = r_{k-1}^0$.
If not, then the edge $(r_{k-1}^0, h_k)$ is labeled $(1, 1)$, as $r_{k-1}^0$ is matched to
the last dummy hospital ($d_{r_{k-1}}^0$) on its preference list (by invariant ($\mathcal{I}_3$)),
and $h_k$ is under-subscribed in $M_{n_1}$.

Thus, in the path $P_{n_1}$, the first two residents are level-$(n_1-1)$,
while the last resident is level-$0$.
Recall that the path $P_{n_1}$ was obtained as an inverse-map of the path $P$ in $G$.
Since the path $P$ contains at most $n_1$ residents (possibly all of the residents in $G$),
the path $P_{n_1}$ also contains at most $n_1$ residents.
From Lemma~\ref{lem:no-jump-transition} we observe that the difference in the levels
of two residents in a sub-path of $P_{n_1}$ can be at most one.
Thus, it must be the case that residents at all the levels $n_1-1$ to $0$
are present in $P_{n_1}$.
However, since there are two residents at level-$(n_1-1)$ (first two residents)
and one resident at level-$0$ (last resident), it is clear that residents at
all levels from $n_1-1$ to $0$ cannot be accommodated in a path containing
at most $n_1$ residents.

This contradicts the existence of such a path $P_{n_1}$ in $G_{n_1}$ which implies
that the assumed augmenting path $P$ with respect to $M$ cannot exist.
This proves that $M = map(M_{n_1})$ is a max-cardinality matching in $G$.
\qed
\end{proof}

We can now conclude that the set $\mathcal{Y}_{M \oplus M^*}$ is a set of
alternating (and not augmenting) paths and alternating cycles.
It remains to show that $M$ is popular
amongst all maximum cardinality feasible matchings in $G$.
Let $M^*$ be any feasible maximum cardinality matching in $G$.
In Lemma~\ref{lem:one-one} we show that if there is an edge $(r, h) \in M^* \setminus M$ labeled $(1,1)$ in $\rho$,
then in  $\rho_{n_1}$, for the hospital $h$,  the level of its unmatched
neighbour (resident) is lower than the level of its matched neighbour (resident).

\begin{lemma}
\label{lem:one-one}
If an edge $(r_{a}, h) \in \rho$ is labeled $(1, 1)$,
then in $\rho_{n_1}$ for the sub-path $\langle r_{a}^{j_{a}}, h, r_b^{j_b} \rangle$
where $M_{n_1} (r_b^{j_b}) = h$,
we have $j_{a} < j_b$.
\end{lemma}
\begin{proof}
Let an edge $e = (r_{a}, h)$ be labeled $(1, 1)$ in $\rho$.
We observe the following about preferences of $r_{a}$ and $h$ in $G$.

\begin{itemize}
\item [($\mathcal{O}_1$)] $r_{a}$ prefers $h$ over $M(r_a)$.
\item [($\mathcal{O}_2$)] $h$ prefers $r_{a}$ over $r_b \in M(h)$,
where $r_b = {\bf corr}(r_{a})$.
\end{itemize}

Consider the edge $e' = (r_{a}^{j_{a}}, h)$ in $G_{n_1}$, as
$h = M_{n_1}(r_b^{j_b})$ it implies $e' \notin M_{n_1}$.
Consider the three cases that can arise depending on the values of $j_a$ and $j_{b}$.

\vspace{0.1in}
\begin{enumerate}
\begin{minipage}{0.3\linewidth}
\item $j_{a} = j_b$
\end{minipage}
\begin{minipage}{0.3\linewidth}
\item $j_{a} > j_b$
\end{minipage}
\begin{minipage}{0.3\linewidth}
\item $j_{a} < j_b$
\end{minipage}
\end{enumerate}

Recall observation $(\mathcal{O}_1)$, and the fact that the residents do not change their
preferences in $G_{n_1}$ w.r.t. the hospitals originally in $G$.
This implies in all the three cases above, the resident $r_{a}^{j_{a}}$
prefers $h$ over $M_{n_1}(r_{a}^{j_{a}})$.
Using $(\mathcal{O}_2)$ and the fact that a hospital $h$ in $G_{n_1}$ prefers level-$p$
residents over level-$q$ residents, when $p > q$, we can conclude the following.
For the cases $(1)$ and $(2)$, hospital $h$ prefers $r_{a}^{j_{a}}$ over $r_b^{j_b}$,
which implies that the pair $(r_{a}^{j_{a}}, h)$ is labeled $(1, 1)$, which
is a blocking pair for $M_{n_1}$ (using Lemma~\ref{lem:blocking-pair-ms}(1)).
This contradicts the stability of $M_{n_1}$.
We therefore conclude that $j_{a} < j_b$.
\qed
\end{proof}

Lemma~\ref{lem:alt-path-resident-maxc} shows that in an alternating path in $\mathcal{Y}_{M \oplus M^*}$
with exactly one endpoint unmatched in $M$ or an alternating cycle, the number
of edges labeled $(1,1)$ cannot exceed the number of edges labeled $(-1,-1)$.

\begin{lemma}
\label{lem:alt-path-resident-maxc}
Let $\rho$ be an alternating path or an alternating cycle in $\mathcal{Y}_{M \oplus M^*}$.
Then the number of edges labeled $(1,1)$ in $\rho$ is at most the number of
edges labeled $(-1,-1)$.
\end{lemma}
\begin{proof}
Depending on the nature of $\rho$ we have three different cases.
\begin{itemize}
\item  $\rho$ is an alternating path which starts with an unmatched resident in $M$.
\item  $\rho$ is an alternating path which starts with a hospital which is under-subscribed in $M$.
\item  $\rho$ is an alternating cycle.
\end{itemize}
The proof idea is similar in all the three cases.
In each of the above, we consider $\rho_{n_1} = map^{-1}(\rho, M_{n_1})$.
For every edge labeled $(1, 1)$ in $\rho$ we show a change (increase / decrease)
in the level of the residents which are neighbours of a particular hospital.
We show that each such change must be complemented with another
change (decrease / increase resp.) in the level of the residents which are neighbours to some other hospital.
Finally, we show that the second type of change translates to a $(-1, -1)$ edge in $\rho$.

Let $\rho$ be an alternating cycle.
Consider a hospital $h_i \in \rho$ for which there is an edge $(r, h)$
labeled $(1, 1)$ incident on it in $\rho$.
Consider the associated hospital $h_i \in \rho_{n_1}$.
W.l.o.g. let $\langle r_{i-1}^{j_{i-1}}, h_i, r_i^{j_i} \rangle$ be a sub-path
of $\rho$ when traversing $\rho_{n_1}$ in counter-clock-wise direction.
By Lemma~\ref{lem:one-one}, we know that the level of the unmatched resident
incident on $h_i$ is lower than the level of the matched resident incident on $h_i$.
Thus there is an {\it increase} in level of residents when at $h_i$
(while traversing $\rho_{n_1}$ in counter-clockwise direction).
This is true for any $h_k \in \rho_{n_1}$ where the associated hospital in $\rho$ has a
$(1, 1)$ edge incident on it.
We now recall from Lemma~\ref{lem:no-jump-transition}(2) that whenever
a hospital $h_k \in \rho_{n_1}$ has a level {\it decrease}, the associated
edge in $\rho$ is labeled $(-1, -1)$.
Furthermore the {\it decrease} in levels at a hospital is gradual.
Thus, it must be the case that the number of $(1, 1)$ edges in $\rho$
is at most the number of $(-1, -1)$ edges in $\rho$.

In case $\rho$ is an alternating path starting at an unmatched resident, we show that in the path $\rho_{n_1}$
the first two residents are level-$(n_1-1)$ residents (see {\bf Claim 1} below for a proof).
Furthermore, consider the first edge $(r, h) \in \rho$ that is labeled $(1, 1)$.
The associated hospital $h_i$ has a {\it increase} in the level of its two neighbouring residents. However, since $\rho_{n_1}$
started with two level-($n_1-1$) residents (which is the highest level possible).
Therefore, there must have been some hospital $h_k$ preceding $h_i$ in $\rho_{n_1}$ which has a {\it decrease} in the
levels of the two neighbours. Using these facts it is easy to prove the following: \\ 
Number of $(1, 1)$ edges in $\rho \le$ Number of increases in $\rho_{n_1} \le$ Number of decreases in $\rho_{n_1} \le$ Number of $(-1,-1)$ edges in $\rho$.

\noindent This completes the proof in case $\rho$ is a path starting at an unmatched resident.

Finally, we are left with the case when $\rho$ is an alternating path which starts with a hospital $h_i$ which is under-subscribed in $M$.
We show that in the associated path $\rho_{n_1}$, the first resident is a level-0 resident (see {\bf Claim 2} below for a proof).
Note that in this case since the path starts at a hospital, whenever we have an edge labeled $(1, 1)$ in $\rho$,
the associated hospital $h_i$ in $\rho_{n_1}$ has an {\it increase} in the levels of the two neighbouring residents.
Now first resident in the path is at the lowest possible level,
it must be the case that there is a hospital $h_k$ preceding $h_i$ in $\rho$ for which there is a {\it decrease} in the level of
the neighbouring residents. Now using arguments similar to those in the case of path starting at an unmatched resident, we conclude
that the number of $(1,1)$ edges in $\rho$ is at most the number of $(-1, 1)$ edges in $\rho$.

{\noindent}{\bf Claim 1:} $\rho = \langle r_0, h_1, r_1, \ldots \rangle$ starts with an unmatched resident.
As $r_0$ is unmatched in $M$, by invariant ($\mathcal{I}_3$), it implies
that for $0 \le i \le n_1-2$, $M_{n_1}(r_0^i) = d_{r}^i$, and $r_0^{n_1-1}$ is unmatched in $M_{n_1}$.
Therefore the first resident $r_0^{j_0}$ on the path $\rho_{n_1}$ is a level-$(n_1-1)$
resident, that is $j_0 = n_1 - 1$.
Furthermore, the second resident on the path $r_1$ has to be a level-$(n_1-1)$ resident.
If not, then as $r_0^{n_1-1}$ is unmatched in $M_{n_1}$ and $h_1$ prefers a level-$(n_1-1)$ resident
to a level-$v$ resident ($v < n_1-1$), the edge $(r_0^{n_1-1}, h_1)$ will be labeled $(1, 1)$,
and thus blocks $M_{n_1}$ contradicting its stability (using Lemma~\ref{lem:blocking-pair-ms}(1)).

{\noindent}{\bf Claim 2:} $\rho = \langle h_0, r_1, h_1, \ldots \rangle$ starts with an under-subscribed hospital.
The first resident $r_1^{j_1}$ on the path $\rho_{n_1}$ has to be a level-$0$
resident, that is $j_1 = 0$.
If not, i.e. if $j_1 = 1$, then $d_{r_0}^0 = M_{n_1}(r_1^0)$
(by invariant ($\mathcal{I}_3$)) $d_{r_0}^0$ in $M_{n_1}$, and prefers $h_1$ to $d_{r_0}^0$.
The hospital $h_1$ on the other hand is under-subscribed in $M_{n_1}$ and prefers
being matched to $r_1^0$ in $M_{n_1}$.
Thus, the edge $(r_1^0, h_1)$ is labeled $(1, 1)$, and blocks $M_{n_1}$ contradicting
its stability (using Lemma~\ref{lem:blocking-pair-ms}(2)).
\qed
\end{proof}

Thus, we get the following theorem:
\begin{theorem}
\label{thm:pop2}
Let $M = map(M_{n_1})$ where $M_{n_1}$ is a stable matching in $G_{n_1}$.
Then $M$ is a popular matching amongst all maximum cardinality matchings in $G$.
\end{theorem}

\noindent {\bf Discussion:} A natural question is to consider popular matchings
in LCSM instances. 
An LCSM instance need not admit a stable matching.
However we claim that restricted to LCSM instances which admit a stable matching, our
results hold without any modification. To obtain the result, we claim that Lemma~\ref{lem:rho_xor_feasible_g} holds  in the presence
of lower quotas on classes. Additionally, if the given LCSM instance $G$ 
admits a stable matching, the graph $G_s$ for $s = 1, \ldots, n_1$ also admits a stable matching.
We thank Prajakta Nimbhorkar for pointing this to us.

\noindent {\bf Acknowledgement:} We thank the anonymous reviewers whose comments have improved the presentation.

%% file: lemma-no-jump.tex
\begin{lemma}
\label{lem:no-jump-transition}
Let $\rho_{n_1}$ be an alternating path or an alternating cycle in $G_{n_1}$
and let $h$ be a  hospital which has degree two in $\rho_{n_1}$.
Let $\langle r_a^{j_a}, h, r_b^{j_b} \rangle$ be the sub-path containing $h$ where
$M(r_b^{j_b}) = h$. If $j_a > j_b$ , we claim the following:
\begin{enumerate}
\item  $j_a = j_b + 1$.
\item  The associated edge $(r_a, h) \in \rho$ is labeled $(-1, -1)$.
\end{enumerate}
\end{lemma}
\begin{proof}
We first prove  that $j_a = j_b + 1$.
For contradiction, assume that $j_a > j_b + 1$.
Observe that $h$ prefers all the level-$j_a$ residents over
any level-$j_b$ resident.
We consider the edge $e = (r_a, h)$ in the graph $G$.
We claim that the label for the edge $e$ cannot be $(1, 1)$ or $(1, -1)$, otherwise
the edge $(r_a^{j_a}, h)$ is labeled $(1, 1)$ in $G_{n_1}$
as the residents do not change their votes.
Similarly, we claim that the label for the edge $e$ cannot be $(-1, 1)$ or $(-1, -1)$,
as $r_a^{j_a-1}$ is matched in $M_{n_1}$ to the last dummy on its
preference list, $d_{r_a}^{j_a-1} = M_{n_1}(r_a^{j_a-1})$
(by invariant $(\mathcal{I}_3)$), and prefers $h$ to $d_{r_a}^{j_a-1}$, and $h$
prefers all the level-$(j_a-1)$ residents over any level-$j_b$ resident.
In this case the edge $(r_a^{j_a-1}, h)$ is labeled $(1, 1)$ in $G_{n_1}$,
and thus blocks $M_{n_1}$ (by Lemma~\ref{lem:blocking-pair-ms}(2)).

To prove part (b), we assume $j_a = j_b+1$.
We enumerate the possible labels for the edge $e = (r_a, h)$ in $G$.
\begin{itemize}
\item If $e$ is labeled $(1,1)$ or $(1,-1)$, then the edge
    $(r_a^{j_a}, h)$ is labeled $(1, 1)$, as $r_a^{j_a}$
    prefers $h$ over $M_{n_1}(r_a^{j_a})$, and $h$
    prefers any level-$j_a$ resident over a level-$j_b$ resident.
    Thus, the edge $(r_a^{j_a}, h)$ blocks $M_{n_1}$
    (by Lemma~\ref{lem:blocking-pair-ms}(1)).

\item If $e$ is labeled $(-1,1)$, then the edge $(r_a^{j_b}, h)$
    is labeled $(1, 1)$, as $r_a^{j_b}$ prefers $h$ over
    $d_{r_a}^{j_b} = M_{n_1}(r_a^{j_b})$, and $h$ prefers $r_a^{j_b}$
    over $r_b^{j_b}$ according to its preference list.
    Thus, the edge $(r_a^{j_b}, h)$ blocks $M_{n_1}$
    (by Lemma~\ref{lem:blocking-pair-ms}(2)).
\end{itemize}
Thus, the only possible label for the edge $(r_{a}, h)$ is $(-1,-1)$.
\qed
\end{proof}

%% file: appendix_reduction_pcsm_spa.tex
\subsection{Reduction from a PCSM$^+$ instance to an SPA instance}
\label{appendix:reduction_pcsm_spa}
An instance of SPA~\cite{Abraham2007} consists of students, projects and lecturers.
Each lecturer has an upper bound on the maximum number of students that he/she is willing to advise.
Each project has an upper bound on the number of students it can accommodate.
Each project is owned by exactly one lecturer.
Each student has a preference ordering over a subset of the projects,
and each lecturer has a preference over the students.

We detail on the reduction from PCSM$^+$ instance to an SPA instance here.
For a resident $r$ in the PCSM$^+$ instance, a corresponding student $s_r$ is
introduced in the SPA instance.
For each hospital $h$, a lecturer $l_h$ with capacity $q(h)$
is added in the SPA instance.
For each class $C^h_j$ in the classification provided by a hospital $h$,
a project $p_j$ is associated with the lecturer $l_h$,
and the upper-bound of $p_j$ is equal to $q(C^h_j)$.
The preference list of $l_h$ is obtained from its corresponding hospital $h$.
If the resident $r$ is the $k$-th most preferred resident in the preference list
of $h$, then the student $s_r$ is the $k$-th most preferred student in the
preference list of $l_h$.
Similarly, the preference list of a student $s_r$ is created from its
corresponding resident $r$.
Let $C^h_j$ be the class that the resident $r$ appears in the classification
provided by the $k$-th most preferred hospital in its preference list,
then $p_j$ is the $k$-th most preferred project in the preference list for $s_r$.
As the classifications associated with every hospital in the PCSM$^+$ instance
are a partition over its preference list, there is no ambiguity in describing
the preference of the students.

%% file: csr.bbl
\begin{thebibliography}{10}

\bibitem{AIKM07}
D.~J. Abraham, R.~W. Irving, T.~Kavitha, and K.~Mehlhorn.
\newblock Popular {M}atchings.
\newblock {\em SIAM Journal on Computing}, 37(4):1030--1045, 2007.

\bibitem{Abraham2007}
D.~J. Abraham, R.~W. Irving, and D.~F. Manlove.
\newblock Two {A}lgorithms for the {S}tudent-{P}roject {A}llocation {P}roblem.
\newblock {\em J. of Discrete Algorithms}, 5(1):73--90, 2007.

\bibitem{BiroMM10}
P.~Bir{\'{o}}, D.~Manlove, and S.~Mittal.
\newblock Size {V}ersus {S}tability in the {M}arriage {P}roblem.
\newblock {\em Theoretical Computer Science}, 411(16-18):1828--1841, 2010.

\bibitem{BrandlK16}
F.~Brandl and T.~Kavitha.
\newblock Popular {M}atchings with {M}ultiple {P}artners.
\newblock {\em CoRR}, abs/1609.07531, 2016.

\bibitem{CsehK16}
{\'{A}}.~Cseh and T.~Kavitha.
\newblock {P}opular {E}dges and {D}ominant {M}atchings.
\newblock In {\em Proceedings of the Eighteenth Conference on Integer
  Programming and Combinatorial Optimization}, pages 138--151, 2016.

\bibitem{FleinerK2012}
T.~Fleiner and N.~Kamiyama.
\newblock A {M}atroid {A}pproach to {S}table {M}atchings with {L}ower {Q}uotas.
\newblock In {\em Proceedings of the Twenty-third Annual ACM-SIAM Symposium on
  Discrete Algorithms}, pages 135--142, 2012.

\bibitem{GS62}
D.~Gale and L.~Shapley.
\newblock {C}ollege {A}dmissions and the {S}tability of {M}arriage.
\newblock {\em American Mathematical Monthly}, 69:9--14, 1962.

\bibitem{Gar75}
P.~G\"{a}rdenfors.
\newblock Match {M}aking: assignments based on bilateral preferences.
\newblock {\em Behavioural Sciences}, 20:166--173, 1975.

\bibitem{GI89}
D.~Gusfield and R.~W. Irving.
\newblock {\em The Stable Marriage Problem: Structure and Algorithms}.
\newblock MIT Press, 1989.

\bibitem{Huang2010}
C.-C. Huang.
\newblock Classified {S}table {M}atching.
\newblock In {\em Proceedings of the Twenty-First Annual {ACM-SIAM} Symposium
  on Discrete Algorithms}, pages 1235--1253, 2010.

\bibitem{HuangK11}
C.-C. Huang and T.~Kavitha.
\newblock {P}opular {M}atchings in the {S}table {M}arriage {P}roblem.
\newblock In {\em Proceedings of 38th International Colloquium on Automata,
  Languages and Programming}, pages 666--677, 2011.

\bibitem{Kamiyama16}
N.~Kamiyama.
\newblock Popular {M}atchings with {T}wo-{S}ided {P}reference {L}ists and
  {M}atroid {C}onstraints.
\newblock {\em Technical Report MI 2016-13}, 2016.

\bibitem{Kavitha14}
T.~Kavitha.
\newblock A {S}ize-{P}opularity {T}radeoff in the {S}table {M}arriage
  {P}roblem.
\newblock {\em {SIAM} Journal on Computing}, 43(1):52--71, 2014.

\bibitem{Kiraly11}
Z.~Kir{\'{a}}ly.
\newblock Better and {S}impler {A}pproximation {A}lgorithms for the {S}table
  {M}arriage {P}roblem.
\newblock {\em Algorithmica}, 60(1):3--20, 2011.

\end{thebibliography}
